\documentclass[11pt]{article}

\usepackage[utf8]{inputenc}
\usepackage{authblk}

\usepackage{amsfonts,amsmath,amssymb,mathtools,amsthm}
\usepackage{graphicx}
\usepackage{stmaryrd}
\usepackage{thmtools}
\DeclareGraphicsExtensions{.eps}
\usepackage{float}

\usepackage{xcolor}
\definecolor{linkcolour}{rgb}{0.15,0.15,0.55}
\definecolor{urlcolour}{rgb}{0.15,0.15,0.55}
\definecolor{citecolour}{rgb}{0.15,0.15,0.55}

\usepackage[linktoc=all,backref=page]{hyperref}
	\hypersetup{
		colorlinks = true,
		linkcolor = linkcolour,
		urlcolor = citecolour,
		citecolor = citecolour,
		linktoc = all,
		hypertexnames = false,
		unicode = true,
		bookmarksnumbered = false,
		pdfmenubar = true,
		pdftoolbar = true}

\usepackage
	[top=2.5cm,
	bottom=2.5cm,
	left=2.5cm,
	right=2.5cm,
	a4paper]
	{geometry}
		\usepackage[margin=2cm]{caption}
\renewcommand{\theequation}{\thesection.\arabic{equation}}

\newcommand\encadremath[1]{\vbox{\hrule\hbox{\vrule\kern8pt
\vbox{\kern8pt \hbox{$\displaystyle #1$}\kern8pt}
\kern8pt\vrule}\hrule}}
\def\enca#1{\vbox{\hrule\hbox{
\vrule\kern8pt\vbox{\kern8pt \hbox{$\displaystyle #1$}
\kern8pt} \kern8pt\vrule}\hrule}}

\newcommand\framefig[1]{
\begin{figure}[bth]
\hrule\hbox{\vrule\kern8pt
\vbox{\kern8pt \vbox{
\begin{center}
{#1}
\end{center}
}\kern8pt}
\kern8pt\vrule}\hrule
\end{figure}
}

\newcommand\figureframex[3]{
\begin{figure}[bth]
\hrule\hbox{\vrule\kern8pt
\vbox{\kern8pt \vbox{
\begin{center}
{\mbox{\epsfxsize=#1.truecm\epsfbox{#2}}}
\end{center}
\caption{#3}
}\kern8pt}
\kern8pt\vrule}\hrule
\end{figure}
}
\newcommand\figureframey[3]{
\begin{figure}[bth]
\hrule\hbox{\vrule\kern8pt
\vbox{\kern8pt \vbox{
\begin{center}
{\mbox{\epsfysize=#1.truecm\epsfbox{#2}}}
\end{center}
\caption{#3}
}\kern8pt}
\kern8pt\vrule}\hrule
\end{figure}
}
 \makeatother

 \usepackage{mathdots}

\renewcommand{\thesection}{\arabic{section}}
\renewcommand{\theequation}{\arabic{section}-\arabic{equation}}
\makeatletter
\@addtoreset{equation}{section}
\makeatother
\newtheorem{theorem}{Theorem}[section]

\newtheorem{proposition}{Proposition}[section]
\newtheorem{lemma}{Lemma}[section]
\newtheorem{corollary}{Corollary}[section]

\theoremstyle{definition}
\newtheorem{remark}{Remark}[section]
\newtheorem{definition}{Definition}[section]

\def\br{\begin{remark}\rm\small}
\def\er{\end{remark}}
\def\bt{\begin{theorem}}
\def\et{\end{theorem}}
\def\bd{\begin{definition}}
\def\ed{\end{definition}}
\def\bp{\begin{proposition}}
\def\ep{\end{proposition}}
\def\bl{\begin{lemma}}
\def\el{\end{lemma}}
\def\bc{\begin{corollary}}
\def\ec{\end{corollary}}
\def\beaq{\begin{eqnarray}}
\def\eeaq{\end{eqnarray}}

\theoremstyle{definition}

\newcommand{\be}{\begin{equation}}
\newcommand{\ee}{\end{equation}}
\newcommand{\beq}{\begin{equation}}
\newcommand{\eeq}{\end{equation}}
\newcommand{\bea}{\begin{eqnarray}}
\newcommand{\eea}{\end{eqnarray}}
\newcommand{\beqq}{\begin{equation*}}
\newcommand{\eeqq}{\end{equation*}}
\newcommand{\beaa}{\begin{eqnarray*}}
\newcommand{\eeaa}{\end{eqnarray*}}

\newcommand{\Tr}{{\operatorname {Tr}}}

\newcommand{\CC}{{\mathbb C}}

\newcommand{\ZZ}{{\mathbb Z}}

\DeclareMathOperator{\PsiDO}{\Psi\text{DO}}

\newcommand{\td}{\tilde}

\newcommand\blfootnote[1]{%
  \begingroup
  \renewcommand\thefootnote{}\footnote{#1}%
  \addtocounter{footnote}{-1}%
  \endgroup
}
\newtheorem{assumption}{Assumption}
\newcommand{\Res}{\mathop{\,\rm Res\,}}
\title{\bf{The Painlev\'{e} I hierarchy: Correspondence between the isomonodromic approach and the minimal models of the KP hierarchy}}

\author{$_{1}$Mohamad Alameddine\footnote{Universit\'{e} Jean Monnet Saint-\'{E}tienne, CNRS, Institut Camille Jordan UMR 5208, Les Forges 2, 20 Rue du Dr Annino, 42000 Saint-Etienne, France. Section de Math\'ematiques, Universit\'e de Gen\`eve, Rue du Conseil-G\'en\'eral 7-9, 1205 Gen\`eve, Switzerland. }\,\,,
$_{2}$Nathan Hayford\footnote{Department of Mathematics, KTH Royal Institute of Technology, Lindstedtsv\"{a}gen 25, Stockholm, Sweden}\,\,,
$_{3}$Olivier Marchal\footnote{Universit\'{e} Jean Monnet Saint-\'{E}tienne, CNRS, Institut Camille Jordan UMR 5208, Institut Universitaire de France, Les Forges 2, 20 Rue du Dr Annino, 42000 Saint-Etienne, France.}
}

\date{\vspace{-5ex}}

\begin{document}

\maketitle

\textbf{Abstract}: Two approaches to the Painlev\'{e} I hierarchy are discussed: the isomonodromic construction based on meromorphic connections, and the minimal models construction based on a reduction of the KP hierarchy. An explicit correspondence between both formalisms is established, identifying these setups explicitly. In particular, this yields new expressions for the Lax matrices and Hamiltonians. 

\blfootnote{\textit{Email Addresses:}
$_{1}$\textsf{Mohamad.Alameddine@unige.ch},
$_{2}$\textsf{nhayford@kth.se},  
$_{3}$\textsf{olivier.marchal@univ-st-etienne.fr}}

\tableofcontents

\section{Introduction}

The primary subject of this work is the Painlev\'e I (PI) hierarchy, a sequence of nonlinear ordinary differential equations whose first member is the Painlev\'e I equation. The PI hierarchy appears in field theories of closed strings \cite{Brezin:1990rb,DS,GM1,GM2}, critical edge transitions in random matrix theory \cite{ClaeysGrava,ClaeysItsKrasovsky,ClaeysVanlessen,DK0}, and integrable regularizations of dispersionless wave equations \cite{DOUGLAS1990176,Dubrovin-06,Dubrovin-Grava-Klein}. It also connects to moduli spaces of connections and character varieties \cite{BoalchYamakawa} through the irregular Riemann--Hilbert correspondence. Painlev\'e hierarchies are closely related to reductions of integrable PDE hierarchies, including KdV and mKdV; see, for example, \cite{FlaschkaNewell1980,Joshi1995,GordoaPickering1999}. Their Lax representations and Hamiltonian structures have been studied from several perspectives \cite{Takasaki,mazzocco2007hamiltonian,TarriconeThesis2022}. The geometric construction in \cite{MarchalP1Hierarchy}, together with the broader framework of \cite{marchal2024hamiltonianrepresentationisomonodromicdeformations}, provides the isomonodromic formulation used here.
In this work, we focus on two different approaches that generate the PI hierarchy: the geometric isomonodromic deformations approach of \cite{MarchalP1Hierarchy} and the string equation/minimal models approach detailed in \cite{Takasaki} arising from a specific reduction of the KP hierarchy.

\medskip

The \textbf{\textit{geometric isomonodromy approach}} developed in \cite{MarchalP1Hierarchy} for the PI hierarchy relies on the study of the space of twisted rank$-2$ meromorphic connections over the projective line $\mathbb{P}^1$ having one ramified pole at infinity of arbitrary order $r_\infty\geq 3$. The starting point (inspired from the non-ramified case of \cite{JimboMiwaUeno,Boalch2001,marchal2024hamiltonianrepresentationisomonodromicdeformations}) is the Fabry--Hukuhara--Turrittin--Levelt theorem providing a formal normal form after ramification at the pole. This formal normal form offers an extension of the space of deformations beyond the space of complex structures of the Riemann surface by including the irregular type of the monodromy considered, generalizing the symplectic structure of Atiyah and Bott \cite{AtiyahBott} beyond the tame case \cite{Boalch2001}. Considering monodromy preserving deformations of the base, one obtains an explicit Hamiltonian formulation of the chosen Darboux coordinates. It is well known that the Painlev\'{e} equations admit a Hamiltonian formulation \cite{Malmquist1922,Okamoto1986Iso} and in \cite{MarchalP1Hierarchy} the Hamiltonians for the PI hierarchy were given explicitly. This was achieved after the reduction of the tangent space of deformations over the base, providing a compatible nonautonomous Hamiltonian system.      

\medskip

The \textbf{\textit{$\mathbf{(2,2g+1)}$ minimal model approach}} to the PI hierarchy, also known as the \textbf{\textit{$\mathbf{(2,2g+1)}$ string equation approach}}, is based on the fact that this hierarchy can be characterized as a reduction of the KP (more specifically here, KdV) hierarchy. This is also the case for string equations of all types \cite{Dubrovin1976,Krichever1977}. This observation was made first for the simplest equation in the hierarchy, the Painlev\'{e} I equation, which characterizes the $(2,3)$ minimal model. Higher order analogues given by an odd integer $2g+1$ constitute the higher elements of the hierarchy. The string equations admit a maximally twisted Lax representation (the coalescing of all the eigenvalues) \cite{Segal1985,sato1983soliton,Date:1982tj}, and in rank$-2$ (the case considered in this work) this Lax representation gives the PI hierarchy through the Lenard recursion operators. This Lax representation admits many similarities with the Mumford system \cite{Mumford2007}, which motivated researchers to apply its Poisson bracket and ``spectral Darboux coordinates" to the PI hierarchy, building all the ingredients needed for considering isomonodromic deformations and thus obtaining the underlying Hamiltonian structure. In this way, one obtains a Hamiltonian for the PI hierarchy, modulo some correction terms which were made explicit in \cite{Takasaki}. 

\medskip

The two approaches to the PI hierarchy are manifestations of two different perspectives. The isomonodromic approach is built on the symplectic nature of isomonodromy systems, and the Darboux coordinates in this setup are motivated by the historical idea of apparent singularities arising when considering the oper gauge. On the other hand, the $(2,2g+1)$ string equation approach relies on an extension of the Mumford system (an isospectral system), and its well known structure. Both approaches provide a Hamiltonian structure for the PI hierarchy and thus a natural question arises: are both approaches giving the same hierarchy? If the answer is in the positive, how could we identify both formalisms and the different perspectives considered? The primary goal of the present work is to address these questions. 

\medskip

The main result of this paper is a positive answer to the above questions: one can identify both formalisms and build an explicit correspondence between these perspectives. Let us note that for this to hold, one needs to identify the symplectic structures of both formalisms. This is accomplished through the identification of the \textbf{\textit{times}}, Lax matrices and the \textbf{\textit{Hamiltonians}} (through the oper Darboux coordinates). This is discussed in \autoref{sec5}. In particular, the correspondence between the base of times (and the tangent space) giving the Hamiltonian flows is done in \autoref{TheoIdentificationTimes}. The time-correspondence is manifested by a dilatation of the base and consequently a dilatation of the Hamiltonian flows of the tangent space of deformations. This result is complemented by the identification of the Hamiltonian structure given in \autoref{TheoIdentificationHamiltonians}, whose proof is based on the identification of the oper and spectral invariants done in \autoref{TheoHinftyI}. It turns out that the correspondence of both setups allows one to identify the Lax formulation of both sides obtaining a full correspondence between the two rank-$2$ Lax representations (\autoref{TheoIdentificationDarbouxCoordinates}). In other words, the Lax representation of the geometric isomonodromic approach is equivalent to the one used on the minimal model side which yields the string equation. This result also extends to an identification of the apparent singularities with the coordinates used on the minimal model side and allows us to generalize the symmetric coordinates of \cite{mazzocco2007hamiltonian} (also used in \cite{MarchalP1Hierarchy}) to the minimal model, an issue left open in \cite{Takasaki}.    

\medskip

The article is organized as follows. In \autoref{sec3} the twisted isomonodromic approach leading to the PI hierarchy is presented based on results of \cite{MarchalP1Hierarchy}. In \autoref{sec4}, the minimal model approach is reviewed with notations mostly based on \cite{Takasaki}. The identification of both formalisms is the main topic of \autoref{sec5} where several theorems associated to different parts of the correspondence are discussed and proved. Furthermore, the extension of the results of \cite{mazzocco2007hamiltonian} left open in \cite{Takasaki} is presented using the symmetric Darboux coordinates. We leave the concluding remarks and possible generalizations for \autoref{sec6}. The article is meant to be self-contained, providing a short but complete review on the subject. With this in mind, we present an additional \autoref{Appendix-A} explaining the reduction of the KP hierarchy giving the PI hierarchy and \autoref{Appendix-B} explaining the origin of the Poisson bracket that identifies the coordinates on both sides as Darboux coordinates. 

\section*{Statements and declarations}
The authors declare that they have no conflict of interest.
Data sharing is not applicable to this article as no datasets were generated or analyzed during the current study.

\section{Notation: Laurent series and symmetric polynomials} \label{sec2}
In this section, we define and list some of the notations that will be used in this article. 
\subsection{Laurent series: positive and negative parts}
We start with the notation for the positive and negative parts of a Laurent series.
\begin{definition}[Positive and negative parts of a Laurent series]\label{NotationSingularPart}
Let $a\in\mathbb C$. For a function $f(\lambda)$ meromorphic at $\lambda=a$, we denote by $\left[f(\lambda)\right]_{a,-}$ its principal part:
\beq f(\lambda)=\sum_{k=-r}^{\infty} F_k (\lambda-a)^{k} \,\, \Rightarrow\,\, \left[f(\lambda)\right]_{a,-}=\sum_{k=1}^{r} F_{-k} (\lambda-a)^{-k} \eeq
For any function $f(\lambda)$ meromorphic at infinity, we denote $\left[f(\lambda)\right]_{\infty,+}$ the polynomial part at infinity (including the $O(1)$ term):
\beq f(\lambda)=\sum_{k=-r}^{\infty} F_k\lambda^{-k} \,\,\Rightarrow\,\, \left[f(\lambda)\right]_{\infty,+}=\sum_{k=0}^{r} F_{-k}\lambda^{k} \eeq
\end{definition}
Finally, we shall also often need to extract one specific coefficient of a Laurent series, and so we denote
\bea \forall \, k_0\geq -r\,:\,  F_{k_0}&:=& \left(f(\lambda)\right)_{(\lambda-a)^{k_0}} \,\,\text{ for }\,\, f(\lambda)=\sum_{k=-r}^{\infty}F_k (\lambda-a)^k\cr
\forall \, k_0\geq -r\,:\, F_{-k_0}&:=&\left(f(\lambda)\right)_{\lambda^{-k_0}} \,\,\text{ for }\,\, f(\lambda)=\sum_{k=-r}^{\infty}F_{-k} \lambda^{-k}
\eea

\subsection{Symmetric polynomials}
We use the following families of symmetric polynomials.
\begin{definition}[Families of symmetric polynomials]\label{DefSymmetricPoly}
\begin{itemize}
\item Elementary symmetric polynomials are denoted by $\left(e_i(\{x_1,\dots,x_n\})\right)_{i\geq 0}$ with the convention that $e_0(\{x_1,\dots,x_n\})=1$ and $e_k(\{x_1,\dots,x_n\})=0$ if $k>n$. By definition we have:
\beq \label{DefElemetarySymPol} e_k(\{x_1,\dots,x_n\})=\sum_{1\leq i_1<\dots<i_k\leq n} x_{i_1}\dots x_{i_k} \,\,,\,\,  \forall \, k\in \llbracket 1, n\rrbracket\eeq
\item Complete homogeneous symmetric polynomials are denoted by $\left(h_i(\{x_1,\dots,x_n\})\right)_{i\geq 0}$ with the convention that $h_0(\{x_1,\dots,x_n\})=1$. By definition we have:
\beq h_k(\{x_1,\dots,x_n\})=\sum_{1\leq i_1\leq \dots\leq i_k\leq n}x_{i_1}\dots x_{i_k} \,\,,\,\,  \forall \, k\geq 1\eeq
\item $k^{\text{th}}$ symmetric power sum polynomials are denoted by $\left(S_k(\{x_1,\dots,x_n\})\right)_{k\geq 0}$. By definition, we have:
\bea S_0(\{x_1,\dots,x_n\})&=&n\cr
S_k(\{x_1,\dots,x_n\})&=&\sum_{j=1}^n x_j^k \,,\, \forall \, k\geq 1
\eea
\end{itemize}
\end{definition}

Over $\mathbb C$, each of the families $(e_1,\ldots,e_n)$, $(h_1,\ldots,h_n)$ and $(S_1,\ldots,S_n)$ is a set of algebraically independent generators of the ring of symmetric polynomials in $x_1,\ldots,x_n$. Monomials in any one of these sets form a vector-space basis. We also have the relations
\bea \label{SymmPoly} \prod_{j=1}^n (\lambda-x_j)&=&\sum_{k=0}^n (-1)^{n-k} e_{n-k}(\{x_1,\dots,x_n\})\lambda^k=\sum_{k=0}^n (-1)^{k} e_{k}(\{x_1,\dots,x_n\})\lambda^{n-k}\cr
\frac{1}{\underset{j=1}{\overset{n}{\prod}} (\lambda-x_j)}&=& \sum_{k=0}^{\infty}h_k(\{x_1,\dots,x_n\})\lambda^{-n-k}
\eea
The relations between the various sets of symmetric polynomials are given by
\bea \label{Relationhe}h_0(\{x_1,\dots,x_n\})&=&e_0(\{x_1,\dots,x_n\})\cr
h_k(\{x_1,\dots,x_n\})&=&
\sum_{j=1}^k(-1)^j
\sum_{\substack{(b_1,\dots,b_j)\in\{1,\dots,k\}^j\\ b_1+\dots+b_j=k}}
\prod_{m=1}^j(-1)^{b_m}e_{b_m}(\{x_1,\dots,x_n\}),
\qquad k\geq1,
\eea
and $\forall\, m\geq 1$:
\small{\bea&&\label{RelationSe} 
S_m(\{x_1,\dots,x_n\})=(-1)^m m\sum_{k=1}^m \frac{1}{k}\hat{B}_{m,k}(-e_1(\{x_1,\dots,x_n\}),\dots,-e_{m-k+1}(\{x_1,\dots,x_n\}))\cr
&&=(-1)^m m\sum_{\substack{b_1+2b_2+\dots+mb_m=m\\ b_1\geq 0,\dots,b_m\geq 0}} \frac{(-1)^{b_1+\dots+b_m}}{(b_1+\dots+b_m)} \binom{b_1+\dots +b_m}{b_1,\dots,b_m } \prod_{i=1}^m e_i(\{x_1,\dots,x_n\})^{b_i}
\eea}
\normalsize{where} $\left(\hat{B}_{m,k}\right)_{m\geq k\geq 0}$ are the ordinary Bell polynomials. Finally, we also have the identities 
\bea \label{Identitites}
(n-k)e_{k}(\{x_1,\dots,x_n\})&=&\sum_{i=0}^k (-1)^{i} e_{k-i}(\{x_1,\dots,x_n\})S_{i}(\{x_1,\dots,x_n\})\,\,,\,\, \forall \,k\in \llbracket 0, n\rrbracket\cr 
S_{k}(\{x_1,\dots,x_n\})&=&\sum_{i=k-n}^{k-1} (-1)^{k-1+i} e_{k-i}(\{x_1,\dots,x_n\})S_{i}(\{x_1,\dots,x_n\})\,\,,\,\, \forall \,k\geq n
\eea
In particular we have $h_0(\{x_1,\dots,x_n\})=1$ and $h_1(\{x_1,\dots,x_n\})=e_1(\{x_1,\dots,x_n\})$.

\medskip 
Elementary symmetric polynomials also satisfy some useful properties in relation with the present results. However, since we shall not need them explicitly in this paper, we refer to  \cite{MarchalP1Hierarchy} for details.

\section{The Painlev\'{e} I hierarchy as a twisted isomonodromy system} \label{sec3}
\subsection{General geometric construction}
The goal of this section is to review the construction of \cite{MarchalP1Hierarchy} leading to the PI hierarchy obtained from the isomonodromic deformations of twisted meromorphic connections of rank $2$ with a single pole at infinity.

\begin{definition}[Space of meromorphic connections with a pole at infinity]
Let $r_\infty\geq 3$ be a given integer. We consider
\beq
F_{\infty, r_\infty}:= \left\{\hat{L}(\lambda) = \sum_{k=1}^{r_\infty-1} \hat{L}^{[\infty,k]} \lambda^{k-1} \,\,/\,\, \left\{\hat{L}^{[\infty,k]}\right\}_{k=1}^{r_\infty-1} \in \left(\mathfrak{gl}_2(\mathbb{C})\right)^{r_\infty-1}\right\}/{GL}_2 
\eeq
where ${GL}_2$ acts simultaneously by conjugation on all the coefficients $\{\hat{L}^{[\infty,k]}\}_{1\leq k\leq r_\infty-1}$. The corresponding meromorphic connection is defined by
\beq \label{System1} d \hat{\Psi}= \hat{L}(\lambda) d\lambda \hat{\Psi} \,\, \Leftrightarrow \, \, \partial_\lambda \hat{\Psi}= \hat{L}(\lambda) \hat{\Psi}\eeq
where $\hat{\Psi}$ is referred to as the wave matrix and $\hat{L}(\lambda)$ is referred to as the Lax matrix.
\end{definition}

There are several notions of twists in the literature. In this work, the term ``twisted'' is manifested by a ramification of the formal normal solution. The PI hierarchy is obtained from a twisted system and for this, it is necessary to consider the following subset of $F_{\infty,r_\infty}$.

\begin{definition}[Generic ramified stratum at infinity]
For $r_\infty\geq3$, set
\[
 L^\circ(\lambda)=\hat L(\lambda)
       -\frac12\Tr\hat L(\lambda)\,I_2,
 \qquad
 \Delta_{\hat L}(\lambda)=-\det L^\circ(\lambda).
\]
We restrict to the subset
\[
 \hat F_{\infty,r_\infty}
 =\left\{[\hat L]\in F_{\infty,r_\infty}:
 \hat L^{[\infty,r_\infty-1]}\text{ is not diagonalizable},
 \quad \deg\Delta_{\hat L}=2r_\infty-5\right\}.
\]
In particular, the coefficient of $\lambda^{2r_\infty-5}$ in
$\Delta_{\hat L}$ is nonzero.
\end{definition}

The odd degree of $\Delta_{\hat L}$ implies that infinity is a branch
point of the normalized spectral curve. Non-diagonalizability of the
leading coefficient alone would not imply this conclusion. The affine
spectral curve is
\[
 \mathcal S=\{(\lambda,y)\in\mathbb C^2:
                  \det(yI_2-\hat L(\lambda))=0\}.
\]
The corresponding formal normal form defines the irregular times.

\begin{proposition}[Formal normal form at infinity]\label{PropDiago}
Put $z=\lambda^{1/2}$. On the generic ramified stratum, the formal
Hukuhara--Turrittin--Levelt theorem gives an invertible formal Laurent
gauge $G_\infty(z)\in
\mathrm{GL}_2(\mathbb C((z^{-1})))$ for which a formal fundamental
solution $\hat\Psi_{\rm form}$ satisfies
\begin{equation}\label{Psi}
 \Psi_\infty(z):=G_\infty(z)\hat\Psi_{\rm form}(z)
 =\operatorname{diag}\left(
       e^{\Theta_+(z)-\frac12\log z},
       e^{\Theta_-(z)-\frac12\log z}\right),
\end{equation}
where
\[
 \Theta_\pm(z)=
 \sum_{k=1}^{2r_\infty-2}\frac{(\pm1)^k t_{\infty,k}}{k}z^k .
\]
Here the common logarithmic term corresponds to the usual ramified
normalization of the wave matrix. The gauge includes the full formal
series; no convergence or holomorphy of that series is asserted.
With $\partial_\lambda=(2z)^{-1}\partial_z$, the transformed matrix is
\begin{align}\label{LocalDiag}
 L_\infty
 &=G_\infty\hat L G_\infty^{-1}
                  +(\partial_\lambda G_\infty)G_\infty^{-1}\nonumber\\
 &=\operatorname{diag}\left(
       \frac12\sum_{k=1}^{2r_\infty-2}
                     t_{\infty,k}\lambda^{\frac{k}{2}-1}-\frac1{4\lambda},
       \frac12\sum_{k=1}^{2r_\infty-2}
                (-1)^k t_{\infty,k}\lambda^{\frac{k}{2}-1}-\frac1{4\lambda}
                       \right).
\end{align}
These are formal identities. The coefficients $t_{\infty,k}$
parameterize the irregular type; the two choices of $z$ exchange the
two diagonal entries.
\end{proposition}

The gauge relation gives
\[
 \det\hat L
 =\det\left(L_\infty-(\partial_\lambda G_\infty)G_\infty^{-1}\right).
\]
The diagonal entries of the transformed differential system are
different from the algebraic eigenvalues of $\hat L$. The latter have
the expansions
\begin{align}\label{LocalDiagP1}
 y_1(\lambda)
 &=\frac12\sum_{k=1}^{2r_\infty-2}
                  t_{\infty,k}\lambda^{\frac{k}{2}-1}
                     +O\left(\lambda^{-\frac{3}{2}}\right),\nonumber\\
 y_2(\lambda)
 &=\frac12\sum_{k=1}^{2r_\infty-2}
                  (-1)^k t_{\infty,k}\lambda^{\frac{k}{2}-1}
                     +O\left(\lambda^{-\frac{3}{2}}\right).
\end{align}
In particular, the common logarithmic correction of the formal
differential system does not occur in these algebraic eigenvalues.

Fixing the irregular type of $\hat{L}(\lambda)$ does not fix it uniquely. In fact, the space 
\beq
\hat{\mathcal{M}}_{\infty,r_\infty,\mathbf{t}} :=\left\{\hat{L}(\lambda) \in \hat{F}_{\infty,r_\infty}\,\,/\,\, \hat{L}(\lambda) \,\text{ has irregular type } \mathbf{t} \right\}
\eeq
is considered on its smooth stable locus, where it is a symplectic manifold realized as the symplectic quotient of a product of coadjoint orbits \cite{Boalch2012}, of dimension
\beq
\dim \hat{\mathcal{M}}_{\infty,r_\infty,\mathbf{t}} = 2r_\infty-6 = 2g
\eeq
where
\beq
\label{GenusDef} g:= r_\infty-3
\eeq
is the genus of the smooth generic compactification of the spectral
curve. At singular specializations it is the arithmetic genus of the
hyperelliptic model, and need not equal the genus of its normalization.
We work on the smooth symplectic locus and use local Darboux coordinates
$(\mathbf q,\mathbf p)$. The first such coordinates below are obtained
from the apparent singularities in the oper gauge.

\bigskip

The irregular times $\mathbf t$ parameterize the base $\mathbb B$ of
deformations; its tangent bundle $T\mathbb B$ has fiber
$T_{\mathbf t}\mathbb B$ at $\mathbf t$. Isomonodromic deformations
preserve the generalized monodromy data, including the Stokes data,
rather than merely the ordinary monodromy on the finite plane.
A deformation direction is described as follows.

\begin{definition}\label{DefGeneralDeformationsDefinition} The general isomonodromic deformation operator is defined by
\beq \label{GeneralDeformationsDefinition}\mathcal{L}_{\boldsymbol{\alpha}}=\sum_{k=1}^{2r_\infty-2} \alpha_{\infty,k} \partial_{t_{\infty,k}}\eeq
where the vector $\boldsymbol{\alpha}\in \mathbb{C}^{2r_{\infty}-2}=\mathbb{C}^{2g+4}$ is defined by
\beq \boldsymbol{\alpha}= \sum_{k=1}^{2r_\infty-2} \alpha_{\infty,k}\mathbf{e}_{k}\eeq
which is a vector of the tangent space.
\end{definition}

For an isomonodromic deformation, the wave matrix satisfies an auxiliary
equation
\[
 \mathcal L_{\boldsymbol\alpha}[\hat\Psi(\lambda)]
       =\hat A_{\boldsymbol\alpha}(\lambda)\hat\Psi(\lambda).
\]
The auxiliary matrix is polynomial in the present one-pole setting.
Before the reduction of the even times, its degree can be
$r_\infty-1$: the highest even time produces a scalar term proportional
to $\lambda^{r_\infty-1}I_2$. The reduced directions considered below
have degree at most $r_\infty-2$. The compatibility of the system
\beq \partial_\lambda \hat{\Psi}(\lambda)=\hat{L}(\lambda)\hat{\Psi}(\lambda)\,\,,\,\, \mathcal{L}_{\boldsymbol{\alpha}}[\hat{\Psi}(\lambda)]= \hat{A}_{\boldsymbol{\alpha}}(\lambda) \hat{\Psi}(\lambda)\eeq
is non trivial and is equivalent to the so-called compatibility equation (or zero curvature equation)
\beq \mathcal{L}_{\boldsymbol{\alpha}}[\hat{L}]-\partial_\lambda \hat{A}_{\boldsymbol{\alpha}}+[\hat{L},\hat{A}_{\boldsymbol{\alpha}}]=0.
\eeq
Following the steps of \cite{MarchalP1Hierarchy}, the compatibility equation is equivalent to the evolution of a chosen set of Darboux coordinates; these evolutions turn out to be Hamiltonian. We review this construction here and refer to \cite{MarchalP1Hierarchy} for the full statements and proofs. 

\bigskip

At this stage the symplectic phase space has dimension $2g=2r_\infty-6$,
whereas there are $2g+4$ irregular times. The identification of trivial
deformation directions requires more than this dimension count.
The explicit reduction in \cite{MarchalP1Hierarchy} identifies the
directions associated with scalar and affine changes. We use the
following canonical normalization to remove them.

\begin{assumption}[Canonical choice of trivial times: reduction of the deformation space]\label{AssumptionTrivial} Following \cite{MarchalP1Hierarchy} we set  
\begin{itemize}\item all even irregular times to $0$: for all $k\in \llbracket 1, r_\infty-1\rrbracket$: $t_{\infty,2k}=0$.
\item $t_{\infty,2r_\infty-3}=2$ and $t_{\infty,2r_\infty-5}=0$.
\end{itemize}
This assumption is based on two observations: the first one is the manifestation of the fact that the passage from $\mathfrak{gl}_2 \rightarrow \mathfrak{sl}_2$ does not change the symplectic structure \cite{FromGaussToPainleve}. The second is the action of affine transformations of the spectral coordinate, which form the subgroup of $\mathrm{PGL}_2$ preserving infinity which preserves the symplectic structure. Moreover, using the conjugation action of $GL_2(\mathbb{C})$, one could consider a representative $\hat{L}(\lambda)\in \hat{\mathcal{M}}_{\infty,r_\infty,\mathbf{t}}$ normalized such that
\beq \label{L-matrix-normalization}
\hat{L}(\lambda) =\begin{pmatrix} 0 & 0\\ 1& 0\end{pmatrix}\lambda^{r_\infty-2}+ \begin{pmatrix}0&1\\X&0\end{pmatrix}\lambda^{r_\infty-3}+ O\left(\lambda^{r_\infty-4}\right).
\eeq
\end{assumption}

In the rest of the paper \textbf{we always assume that \autoref{AssumptionTrivial} is verified and that $\hat{L}(\lambda)$ is normalized according to \eqref{L-matrix-normalization}.} This implies in particular that $\Tr \,\hat{L}(\lambda)=0$.
For the correspondence with the PI hierarchy we henceforth assume
$r_\infty\geq4$, so that $g=r_\infty-3\geq1$.
This assumption allows for a reduction of \eqref{LocalDiagP1}:

\begin{proposition}[Local expansion of the eigenvalues at infinity]\label{PropLocalExpansionsEigenvalues} Let $(y_1(\lambda),y_2(\lambda))$ be the two eigenvalues of $\hat{L}(\lambda)$, then under \autoref{AssumptionTrivial}, we have
\begin{align}\label{LocalDiagP1Reduced}
 y_1(\lambda)
 &=\lambda^{r_\infty-\frac52}
   +\frac12\sum_{k=0}^{r_\infty-4}
            t_{\infty,2k+1}\lambda^{k-\frac12}
   +O(\lambda^{-3/2}),\nonumber\\
 y_2(\lambda)&=-y_1(\lambda).
\end{align}
    
\end{proposition}

\subsection{Explicit formulas in the oper gauge}
Let us now briefly recall the results of \cite{MarchalP1Hierarchy} in the oper gauge. In this gauge, the appropriate Darboux coordinates are given using the apparent singularities. In particular one defines:

\begin{definition}[Oper Darboux coordinates and apparent singularities]\label{DarbouxDefqp}
We work on the open subset where $\hat L_{1,2}$ has $g$ distinct roots,
with a local ordering $(q_1,\ldots,q_g)$. Equivalently,
$\prod_{i<j}(q_i-q_j)\ne0$. We define these roots by
\beq\label{Conditionqi}
\forall\, i\in \llbracket 1,g\rrbracket \, : \; \left[\hat{L}(q_i)\right]_{1,2} = 0.
\eeq
The conjugate coordinates are obtained by evaluating the entry $\left[\hat{L}(\lambda)\right]_{1,1}$  at $\lambda = q_i$,
\beq \label{Conditionpi}
\forall\, i\in \llbracket 1,g\rrbracket \, : \; p_i:=\left[\hat{L}(q_i)\right]_{1,1}.
\eeq    
We denote $\mathbf{q}:=(q_1,\dots,q_g)$ and $\mathbf{p}:=(p_1,\dots,p_g)$ and these coordinates are referred to as the ``oper Darboux coordinates". 
\end{definition}

The choice of the dual partners is such that for any $i\in \llbracket 1,g\rrbracket$, the pair $(q_i,p_i)$ is a point on the spectral curve, namely $\det(p_i I_2-\hat{L}(q_i))=0$. This set of Darboux coordinates is directly related to the apparent singularities of the system. They appear when performing the gauge transformation seen as a change of trivialization taking the connection to the oper gauge, i.e. the gauge in which the Lax matrix is companion-like and thus reduces to a scalar equation. The terminology ``oper Darboux coordinates" stands for the fact that these coordinates are linked with the oper gauge which corresponds to the scalar reformulation of the connection. For rank two systems, the gauge transformation towards the oper gauge is explicit and given by:

\begin{proposition}[Oper gauge transformation \cite{MarchalP1Hierarchy}]\label{OperGaugeTransfo} Define the gauge transformation
\beq \label{GaugeGexpr}  \Psi(\lambda)=G(\lambda) \hat{\Psi}(\lambda) \,\,\text{with}\,\, G(\lambda)=\begin{pmatrix} 1&0\\ \hat{L}_{1,1}(\lambda)& \hat{L}_{1,2}(\lambda)\end{pmatrix}\eeq
Then $\Psi$ is a solution of the companion-like system
\beq \label{CompanionMatrix}\partial_\lambda \Psi(\lambda)=L(\lambda)\Psi(\lambda)\,\,\text{with}\,\, L(\lambda)=\begin{pmatrix}0&1\\ L_{2,1}(\lambda)&L_{2,2}(\lambda)\end{pmatrix}\eeq 
given by
\bea \label{LInTermsOfTdL} L_{2,1}&=&-\det \hat{L}+\partial_\lambda\hat{L}_{1,1}-\hat{L}_{1,1}\frac{\partial_\lambda\hat{L}_{1,2}}{\hat{L}_{1,2}},\cr
L_{2,2}&=&\Tr\, \hat{L} +\frac{\partial_\lambda\hat{L}_{1,2}}{\hat{L}_{1,2}}.
\eea
\end{proposition}

Note in particular that the first rows of $\Psi$ and $\hat{\Psi}$ are obviously the same: $\Psi_{1,1}=\hat{\Psi}_{1,1}:=\psi_1$ and $\Psi_{1,2}=\hat{\Psi}_{1,2}:=\psi_2$, so that we immediately get
\beq \Psi(\lambda)=\begin{pmatrix}\hat{\Psi}_{1,1}(\lambda)& \hat{\Psi}_{1,2}(\lambda)\\ \partial_\lambda \hat{\Psi}_{1,1}(\lambda)& \partial_\lambda \hat{\Psi}_{1,2}(\lambda) \end{pmatrix}=\begin{pmatrix}\psi_{1}(\lambda)& \psi_2(\lambda)\\ \partial_\lambda \psi_{1}(\lambda)& \partial_\lambda \psi_{2}(\lambda) \end{pmatrix}.\eeq
The companion-like system \eqref{CompanionMatrix} is equivalent to the statement that $\psi_1$ and $\psi_2$ satisfy the following linear ODE:
\beq \left(\left[\partial_{\lambda}\right]^2 -L_{2,2}(\lambda)\partial_\lambda -L_{2,1}(\lambda)\right)\psi_i=0,\eeq 
which is sometimes referred to as the ``quantum curve''.

\bigskip
At a simple root $q_i$ of $\hat L_{1,2}$, the oper matrix $L$ has at most
a simple pole. The polynomial gauge $G$ and a local holomorphic
fundamental solution $\hat\Psi$ show that $\Psi=G\hat\Psi$ has holomorphic
entries and trivial local monodromy there. Its determinant vanishes at
$q_i$, so it is a fundamental matrix only on the punctured
neighborhood. This is the sense in which $q_i$ is an apparent
singularity. Since $\hat L_{1,2}(q_i)=0$ and $\hat L_{1,1}(q_i)=p_i$,
the identity $\det(p_iI_2-\hat L(q_i))=0$ follows directly.

\bigskip
Let us now recall the main results of \cite{MarchalP1Hierarchy} providing the explicit expression of $L(\lambda)$ and of the Hamiltonian evolutions of $(\mathbf{q},\mathbf{p})$ under \autoref{AssumptionTrivial}.

\begin{proposition}[Explicit expression of the gauge matrix $G$ (Prop. $2.2$ of \cite{MarchalP1Hierarchy})]\label{PropGaugeExplicit}Under \autoref{AssumptionTrivial}, the gauge matrix $G$ towards the oper gauge is rewritten in terms of $(\mathbf{q},\mathbf{p};\mathbf{t})$:
\beq G(\lambda)= \begin{pmatrix}1 &0\\ -Q(\lambda)& \underset{j=1}{\overset{g}{\prod}}(\lambda-q_j)\end{pmatrix} \,\,\Leftrightarrow\,\, G(\lambda)^{-1}=\begin{pmatrix}1 &0\\
\frac{Q(\lambda)}{\underset{j=1}{\overset{g}{\prod}}(\lambda-q_j)}& \frac{1}{\underset{j=1}{\overset{g}{\prod}}(\lambda-q_j)}
\end{pmatrix}\eeq
where $Q$ is the Lagrange polynomial
\beq \label{DefQ}  Q(\lambda):= -\sum_{i=1}^g p_i \prod_{j\neq i}\frac{\lambda-q_j}{q_i-q_j} \,\, \text{i.e. } Q(q_i)=-p_i \,,\, \forall \, i\in \llbracket 1,g\rrbracket\eeq
\end{proposition}

\begin{proposition}[Explicit expression of the Lax matrix $L$ in the oper gauge (Prop. $5.1$ of \cite{MarchalP1Hierarchy})]\label{PropHinftyk}Under \autoref{AssumptionTrivial}, the Lax matrix $L$ in the oper gauge is given by
\bea L_{1,1}(\lambda)&=&0 \,\,,\,\,L_{1,2}(\lambda)=1\cr
L_{2,2}(\lambda)&=&\sum_{j=1}^g \frac{1}{\lambda-q_j}\cr
L_{2,1}(\lambda)&=&-P_2(\lambda) +\sum_{k=0}^{r_\infty-4} H_{\infty,k}(\mathbf{q},\mathbf{p},\mathbf{t})\lambda^k - \sum_{j=1}^g\frac{p_j}{\lambda-q_j}
\eea
where $P_2$ is the polynomial in $\lambda$ whose coefficients are given by the irregular times:
\bea \label{ReducedtdP2}P_2(\lambda)&:=&-\lambda^{2r_\infty-5}-\sum_{k=r_\infty-2}^{2r_\infty-7}\left(t_{\infty,2k-2r_\infty+7}
+\frac{1}{4}\sum_{m=k-r_\infty+6}^{r_\infty-3}t_{\infty,2m-1}t_{\infty,2k-2m+5}\right)\lambda^k\cr
&&-\left(t_{\infty,1}+\frac{1}{4}\sum_{m=3}^{r_\infty-3}t_{\infty,2m-1}t_{\infty,2r_\infty-2m-1}\right)\lambda^{r_\infty-3}\cr
&:=&\sum_{k=r_\infty-3}^{2r_\infty-5} P_{\infty,k}^{(2)}\,\lambda^{k}
\eea
and the coefficients $\left(H_{\infty,k}(\mathbf{q},\mathbf{p},\mathbf{t})\right)_{k=0}^{r_\infty-4}$ are given by
\beq \label{DefCi2}
\left(V_\infty\right)^{t}\begin{pmatrix}H_{\infty,0}(\mathbf{q},\mathbf{p},\mathbf{t})\\ \vdots\\ H_{\infty,r_\infty-4}(\mathbf{q},\mathbf{p},\mathbf{t}) \end{pmatrix}=\begin{pmatrix} p_1^2 +P_2(q_1)+ \underset{i\neq 1}{\sum}\frac{p_i-p_1}{q_1-q_i}\\
\vdots\\
p_g^2+P_2(q_g)+ \underset{i\neq g}{\sum}\frac{p_i-p_g}{q_g-q_i}
\end{pmatrix}
\eeq
and $V_\infty$ is the $g\times g$ Vandermonde matrix:
\beq \label{DefVinfty}V_\infty:=\begin{pmatrix}1&1 &\dots &\dots &1\\
q_1& q_2&\dots &\dots& q_{g}\\
\vdots & & & & \vdots\\
\vdots & & & & \vdots\\
q_1^{r_\infty-4}& q_2^{r_\infty-4} &\dots & \dots& q_{g}^{r_\infty-4}\end{pmatrix}\,,\, 
\eeq
\end{proposition}

The Hamiltonian evolutions of $(\mathbf{q},\mathbf{p})$ are given by

\begin{proposition}[Hamiltonian evolutions of $(\mathbf{q},\mathbf{p})$, (Th. $8.1$ of \cite{MarchalP1Hierarchy})]\label{Hamqp}Under \autoref{AssumptionTrivial}, the Hamiltonian evolutions of $(\mathbf{q},\mathbf{p})$ are given by
\beq \label{DefHamReduced} \text{Ham}^{(\boldsymbol{\alpha})}(\mathbf{q},\mathbf{p},\mathbf{t})=\sum_{k=0}^{r_\infty-4} \nu_{\infty,k+1}^{(\boldsymbol{\alpha})}(\mathbf{t})H_{\infty,k}(\mathbf{q},\mathbf{p},\mathbf{t})\eeq
where the time-dependent coefficients $\left(\nu_{\infty,k}^{(\boldsymbol{\alpha})}(\mathbf{t})\right)_{k=1}^{r_\infty-3}$ are given by
\beq \label{RelationNuAlphaInfty} \begin{pmatrix}2&0&\dots&& &0\\
0& 2&0& \ddots& &\vdots\\
t_{\infty,2r_\infty-7}&\ddots&\ddots& \ddots&&\vdots\\
\vdots& \ddots& \ddots&&&\vdots \\
t_{\infty,7}& & \ddots &\ddots&&0 \\
t_{\infty,5}& t_{\infty,7} & \dots & t_{\infty,2r_\infty-7}& 0&2\end{pmatrix}\begin{pmatrix} \nu^{(\boldsymbol{\alpha})}_{\infty,1}(\mathbf{t})\\  \vdots \\  \nu^{(\boldsymbol{\alpha})}_{\infty,r_\infty-3}(\mathbf{t}) \end{pmatrix}=\begin{pmatrix}\frac{2\alpha_{\infty,2r_\infty-7}}{(2r_\infty-7)}  \\ \vdots \\ \frac{2\alpha_{\infty,1}}{1}  \end{pmatrix} 
\eeq
\end{proposition}

\bigskip

Note that formulas for the auxiliary matrix $A_{\boldsymbol{\alpha}}(\lambda)$ in the oper gauge (defined as $\mathcal{L}_{\boldsymbol{\alpha}}[\Psi]=A_{\boldsymbol{\alpha}} \Psi$) are also available in \cite{MarchalP1Hierarchy}. An important property of the expression \eqref{DefHamReduced} is that the Hamiltonians can be seen as time-dependent linear combinations (with factor $\left(\nu_{\infty,k}^{(\boldsymbol{\alpha})}\right)_{1\leq k\leq g}$ depending on the direction of the isomonodromic deformation) of quantities $\left(H_{\infty,k}\right)_{0\leq k\leq g-1}$ that only come from the Lax matrix $\hat{L}$ and thus do not depend on the deformation vector $\boldsymbol{\alpha}$.

\subsection{Explicit expression in the geometric gauge}
Let us mention that one could use the gauge transformation of \autoref{PropGaugeExplicit} and the explicit expression of the Lax matrix $L$ in the oper gauge to obtain some formulas for the Lax matrix $\hat{L}$ in the geometric gauge.

\begin{proposition}[Geometric Lax matrix in terms of the oper Darboux coordinates]\label{hatLqp}Under \autoref{AssumptionTrivial}, the geometric Lax matrix $\hat{L}$ is given in terms of $(\mathbf{t},\mathbf{q},\mathbf{p})$ by:
\bea \label{CheckLEquationsReduced}\hat{L}_{1,1}(\lambda)&=&-Q(\lambda),\cr
\hat{L}_{1,2}(\lambda)&=&\underset{j=1}{\overset{g}{\prod}}(\lambda-q_j),\cr
\hat{L}_{2,2}(\lambda)&=&Q(\lambda),\cr
\hat{L}_{2,1}(\lambda)&=& \frac{ \partial}{\partial \lambda} \bigg(\frac{Q(\lambda)}{\underset{j=1}{\overset{g}{\prod}} (\lambda-q_j)}\bigg) +\frac{L_{2,1}(\lambda) }{\underset{j=1}{\overset{g}{\prod}} (\lambda - q_j)}  - \frac{Q(\lambda)^2}{\underset{j=1}{\overset{g}{\prod}}(\lambda-q_j)}.
\eea 
\end{proposition}

Using the explicit form of the gauge transformation, one has the following:
\bea \label{Equivalence}\det \hat{L}&=&-L_{2,1}+\hat{L}_{1,2}\,\partial_\lambda\left(\frac{ \hat{L}_{1,1}}{\hat{L}_{1,2}}\right)\cr
&=&P_2(\lambda) -\sum_{j=0}^{r_\infty-4}H_{\infty,j}\lambda^j+\sum_{j=1}^{g} \frac{p_j}{\lambda-q_j}+\hat{L}_{1,2}\partial_\lambda\left(\frac{ \hat{L}_{1,1}}{\hat{L}_{1,2}}\right).
\eea
Let us now use \autoref{hatLqp} to express $\hat{L}_{1,2}$ and $\hat{L}_{1,1}$ in \eqref{Equivalence}. We find that
\bea \label{ExpressionDethatL}\det \hat{L}
&=&P_2(\lambda) -\sum_{j=0}^{r_\infty-4}H_{\infty,j}\lambda^j+\sum_{j=1}^{g} \frac{p_j}{\lambda-q_j}+ \left(\prod_{j=1}^g(\lambda-q_j)\right)\partial_\lambda\left(\sum_{i=1}^g\frac{p_i}{\lambda-q_i}\prod_{k\neq i}\frac{1}{(q_i-q_k)}\right)\cr
&=& P_2(\lambda) -\sum_{j=0}^{r_\infty-4}H_{\infty,j}\lambda^j+\sum_{j=1}^{g} \frac{p_j}{\lambda-q_j}-\sum_{i=1}^g\frac{p_i}{(\lambda-q_i)}\prod_{k\neq i}\frac{(\lambda-q_k)}{(q_i-q_k)}.
\eea

Note that this gauge yields a complicated formula for the entry $\hat{L}_{2,1}(\lambda)$. In particular, the expression is non-polynomial in $\lambda$ at first sight. Moreover, from their explicit expression recalled in \autoref{Hamqp}, the Hamiltonians for the oper Darboux coordinates $(\mathbf{q},\mathbf{p})$ are not polynomial. In order to solve these two problems, it is convenient to perform a change of Darboux coordinates.

\begin{definition}[Symmetric Darboux coordinates]\label{DefNewCoord} We define $(\mathbf{Q},\mathbf{P}):=(Q_1,\dots,Q_g,P_1,\dots,P_g)$ using the elementary symmetric polynomials (See \eqref{DefElemetarySymPol}):
\bea Q_i&=&e_i(q_1,\dots,q_g)  \cr
p_i&=&\sum_{k=1}^g P_k \frac{\partial e_k(q_1,\dots,q_g)}{\partial q_i} \,\,,\,\, \forall \, i\in \llbracket 1,g\rrbracket
\eea
We refer to $(Q_1,\dots,Q_g,P_1,\dots,P_g)$ as the symmetric Darboux coordinates.
We set $Q_0=1$ and $Q_j=0$ for $j>g$.
\end{definition}

In \cite{MarchalP1Hierarchy}, it is proved that the change of Darboux coordinates $(\mathbf{q},\mathbf{p})\leftrightarrow (\mathbf{Q},\mathbf{P})$ is symplectic (and obviously time-independent). These symmetric Darboux coordinates are also particularly convenient to express the Lagrange polynomial $Q(\lambda)$ since
\beq \label{QCorollary}Q(\lambda)= -\sum_{i=1}^g p_i \prod_{j\neq i}\frac{\lambda-q_j}{q_i-q_j}=\sum_{j=0}^{g-1} (-1)^{j-1}\left(\sum_{i=j+1}^g P_i Q_{i-j-1}\right)\lambda^j
\eeq
and one can find the explicit expression of $\hat{L}(\lambda)$ using these symmetric Darboux coordinates.

\begin{proposition}[Geometric Lax matrix in terms of symmetric Darboux coordinates]\label{IsoPropSym}
Under \autoref{AssumptionTrivial}, the geometric Lax matrix $\hat{L}$ is given in terms of $(\mathbf{t},\mathbf{Q},\mathbf{P})$ by:
\footnotesize{\bea \hat{L}_{1,1}(\lambda)&=&-\sum_{j=0}^{g-1}(-1)^{j-1}\left(\sum_{i=j+1}^{g} P_i Q_{i-j-1}\right)\lambda^j\cr
\hat{L}_{1,2}(\lambda)&=& \sum_{m=0}^{g}(-1)^{g-m} Q_{g-m} \lambda^m\cr
\hat{L}_{2,2}(\lambda)&=&\sum_{j=0}^{g-1}(-1)^{j-1}\left(\sum_{i=j+1}^{g} P_i Q_{i-j-1}\right)\lambda^j\cr 
\hat{L}_{2,1}(\lambda)&=&-\sum_{i=0}^{g+1}\sum_{j=g+i}^{2g+1}P_{\infty,j}^{(2)}\,h_{j-g-i}(\{\mathbf{q}\})\lambda^i-\cr
&&\sum_{i=0}^{g-2}\left( \sum_{j_1=i+1}^{g-1}\sum_{j_2=g+i-j_1}^{g-1} (-1)^{j_1+j_2}\left(\sum_{i_1=j_1+1}^{g} P_{i_1} Q_{i_1-j_1-1}\right)\left(\sum_{i_2=j_2+1}^{g} P_{i_2} Q_{i_2-j_2-1}\right) h_{j_1+j_2-g-i}(\{\mathbf{q}\})\right)\lambda^i\cr&&
\eea}
\normalsize{where} the complete symmetric functions are defined for all
$m\geq0$ by
\begin{equation}\label{Relationhe2}
 \sum_{m=0}^\infty h_m(\{\mathbf q\})\,z^m
   =\frac{1}{\underset{j=0}{\overset{g}{\sum}}(-1)^jQ_jz^j}.
\end{equation}
Thus $h_0=1$, and
$h_m=\underset{j=1}{\overset{\min(m,g)}{\sum}}(-1)^{j+1}Q_jh_{m-j}$ for $m\geq1$.
We also put $h_m=0$ for $m<0$.
\end{proposition}

The main advantage of the former expressions is that they are clearly polynomials in the Darboux coordinates $(\mathbf{Q},\mathbf{P})$. Since the change for $(\mathbf{q},\mathbf{p})\leftrightarrow(\mathbf{Q},\mathbf{P})$ is time-independent and symplectic, it is easy to obtain the Hamiltonian evolutions of $(\mathbf{Q},\mathbf{P})$ from the Hamiltonian evolutions of $(\mathbf{q},\mathbf{p})$ (a direct replacement inside the formula is sufficient).

\begin{proposition}[Expression of the Hamiltonians in terms of the symmetric Darboux coordinates, Th. $8.1$ of \cite{MarchalP1Hierarchy}]\label{HamQP}Under \autoref{AssumptionTrivial}, the Hamiltonian evolutions of $(\mathbf{Q},\mathbf{P})$ are given by
\bea&&\td{\text{Ham}}^{(\boldsymbol{\alpha})}(\mathbf{Q},\mathbf{P},\mathbf{t})=\cr&&
-\sum_{i=1}^g\nu^{(\boldsymbol{\alpha})}_{\infty,i}\sum_{k=i+1}^g\left((-1)^{i}(g-i)P_k Q_{k-1-i}+\sum_{m=i+1}^{k-1}(-1)^{m}P_k Q_{k-1-m}S_{m-i}(\{\mathbf{q}\})\right)\cr
&&+\sum_{i=1}^g\nu^{(\boldsymbol{\alpha})}_{\infty,i}\sum_{k_1=1}^g\sum_{k_2=1}^g P_{k_1}P_{k_2}\Big[(-1)^{i-1}\sum_{r_1=\text{Max}(0,i-k_2)}^{\text{Min}(k_1-1,i-1)}Q_{k_1-1-r_1}Q_{k_2-i+r_1}\cr
&&+ \displaystyle \sum_{\substack{0\leq r_1\leq k_1-1 \\ 0\leq r_2\leq k_2-1 \\ r_1+r_2\geq g}}(-1)^{r_1+r_2}Q_{k_1-1-r_1}Q_{k_2-1-r_2}\sum_{m=i}^g (-1)^{g-m}Q_{g-m}h_{r_1+r_2+m-i-g+1}(\{\mathbf{q}\})\Big]\cr
&&+\sum_{i=1}^g\nu^{(\boldsymbol{\alpha})}_{\infty,i}\sum_{r=g}^{2r_\infty-5}\sum_{m=i}^{g} (-1)^{g-m} P_{\infty,r}^{(2)} Q_{g-m}h_{r+m-i-g+1}(\{\mathbf{q}\}),
\eea
Here $h_m(\{\mathbf q\})$ is defined for every $m\geq0$ by
\eqref{Relationhe2}, with $h_m=0$ for $m<0$. The power sums are determined
by $S_0=g$ and the generating identity
\begin{equation}\label{RelationSeReduced}
 \sum_{m=1}^\infty\frac{S_m(\{\mathbf q\})}{m}z^m
       =-\log\left(\sum_{j=0}^g(-1)^jQ_jz^j\right).
\end{equation}
These definitions cover every index occurring in the displayed
Hamiltonians.
\end{proposition}

Explicit expressions for the auxiliary matrix $\hat{A}_{\boldsymbol{\alpha}}(\lambda)$ in terms of $(\mathbf t,\mathbf{q},\mathbf{p})$ or $(\mathbf t,\mathbf{Q},\mathbf{P})$ are also available in \cite{MarchalP1Hierarchy}.

\section{The Painlev\'{e} I hierarchy from $(2,2g+1)$ minimal models} \label{sec4}
In this section, we review the construction of the PI hierarchy from the $(p,q)$ minimal models with $p=2$ and $q=2g+1$, also known in the literature as the string equation method. In practice, it corresponds to a specific reduction of the KP hierarchy with additional constraints. In particular, we would like to prove that the $(2,2g+1)$ minimal models are indeed equivalent to the isomonodromic deformation approach presented in the previous section with $g=r_\infty-3$. Our review on $(2,2g+1)$ minimal models is mostly based on \cite{Takasaki} regarding notations and content. 

\subsection{Definition from compatible PDEs}

Let us define the PI hierarchy via the so-called \textit{KP hierarchy approach}. This corresponds to taking a specific reduction of the general $KP$ hierarchy. In fact, the PI hierarchy is a specific reduction of the so-called \textit{KdV hierarchy}, which is itself already a particular reduction of the KP hierarchy. For completeness, the definition of the KdV hierarchy using pseudodifferential operators and Lenard recursion is recalled in \autoref{Appendix-A}. We first define a sequence of differential polynomials $\left(R_{2k-1}[u]\right)_{k\geq 0}$ of a function $u:=u(x)$ via the following recursion relation.

\begin{definition}[Definition of the differential polynomials $\left(R_{2\ell-1}\right)_{\ell \geq 0} $ ]\label{DefinitionRks} Set 
\bea \label{DefRs} R_{-1}[u]&:=&1 \text{ and }\,\,\,
\forall \, \ell\geq 0\,:\,\cr
R_{2\ell +1}[u]&=&-\frac{1}{8}\sum_{\ell_1=1}^{\ell -1}(\partial_x R_{2\ell_1-1}[u])(\partial_x R_{2\ell-2\ell_1-1}[u])+\frac{u}{2}\sum_{\ell_1=0}^{\ell}R_{2\ell_1-1}[u]R_{2\ell-2\ell_1-1}[u]\cr
&&+\frac{1}{4}\sum_{\ell_1=1}^{\ell}R_{2\ell-2\ell_1-1}[u]\frac{\partial^2}{\partial x^2}\left(R_{2\ell_1-1}[u]\right)-\frac{1}{2}\sum_{\ell_1=0}^{\ell-1}R_{2\ell_1+1}[u]R_{2\ell-2\ell_1-1}[u],
\eea
which is equivalent for the formal generating series $B(\lambda)$:
\begin{equation}
     B(\lambda) := \sum_{\ell=0}^{\infty}R_{2\ell-1}[u]\lambda^{-\ell}
\end{equation}
to the relation
\begin{equation}
            \frac{1}{4}\left(\frac{\partial}{\partial x}B(\lambda)\right)^2 + \left(\left(\lambda-u(x)\right)B(\lambda)-\frac{1}{2}\frac{\partial^2}{\partial x^2}B(\lambda)\right)B(\lambda) = \lambda.
        \end{equation}
with the initial condition $B(\lambda)\overset{\lambda\to \infty}{=} 1+O\left(\lambda^{-1}\right)$,
\end{definition}

Note that \autoref{DefinitionRks} implies that the following recursive relation, known as the Lenard recursion, holds:
    \begin{equation}\label{Lenard-recursion}
      \forall\,\ell\geq 0\,:\,  \frac{\partial}{\partial x} R_{2\ell+1}[u] = \left(\frac{1}{4}\frac{\partial^3}{\partial x^3} + u\frac{\partial}{\partial x} + \frac{1}{2}u_{x}
      \right)R_{2\ell-1}[u], \qquad\qquad R_{-1}[u] :=1.
    \end{equation}

\begin{remark}\label{RemarkKdV}The fact that \eqref{Lenard-recursion} can be integrated and gives a differential polynomial in $u$ follows either from the fact that the definition implies the Lenard recursion or from the reduction of the KdV hierarchy. Indeed, the differential polynomials $\left(R_{2\ell-1}[u]\right)_{\ell \geq 0}$ appear as residues (i.e. coefficient of $\partial^{-1}$ in the pseudo-differential expansion) of certain pseudo-differential operators. To be more specific, if we define the pseudo-differential operator $Q^{1/2} := \left(\partial^2+u\right)^{1/2} = \partial +\frac{1}{2}u\partial^{-1} - \frac{1}{4}u_x\partial^{-2} + O\left(\partial^{-3}\right)$, then
    \begin{equation}
\forall\, \ell \geq 0 \,:\,\,  R_{2\ell+1}[u] =\Res_{\partial} Q^{\ell+\frac12}.
    \end{equation}
Since the pseudo-differential expansion of $Q^{1/2}$ (and thus any power of $Q^{1/2}$) has coefficients in the ring of differential polynomials in $u$, $R_{2\ell+1}[u]$ is itself a differential polynomial in $u$. Thus, since the Lenard recursion relation \eqref{Lenard-recursion} is a relation between $R_{2\ell+1}$ and $R_{2\ell-1}$, the fact that it can be integrated follows directly from the a priori well-definedness of $R_{2\ell+1}$ and $R_{2\ell-1}$ as differential polynomials in $u$ despite not being obvious from the recursion itself. We refer to \autoref{Appendix-A} for a short review of the derivation of the PI hierarchy from the pseudo-differential operators approach.
\end{remark}    
 
The first few differential polynomials are 
\bea\label{FirstRk} R_1[u]&=&\frac{1}{2}u\cr
R_3[u]&=&\frac{1}{8}u_{xx} + \frac{3}{8}u^2\cr
R_5[u]&=&\frac{1}{32}u_{xxxx}+\frac{5}{16}uu_{xx}+\frac{5}{32}(u_x)^2+\frac{5}{16}u^3\cr
R_7[u]&=&\frac{1}{128}u_{xxxxxx}+\frac{7}{64}u u_{xxxx}+\frac{21}{128}(u_{xx})^2+\frac{7}{32}u_xu_{xxx}\cr&&
+\frac{35}{64}u^2u_{xx}+\frac{35}{64}u(u_x)^2+\frac{35}{128}u^4
\eea
where the notation $u_x$ stands for $\partial_x u$.

\bigskip

We now define the PI hierarchy as follows.

    \begin{definition}\label{DefP1HierarchyReductionKP}
        Let $g\in \mathbb{N}\setminus\{0\}$. The $g^{\text{th}}$ member of the PI hierarchy is an ODE of order $2g$ of a function $u:= u(x)$, supplemented by $g-1$ additional ``compatible" flows $\left(s_{2k+1}\right)_{1\leq k\leq g-1}$ 
        which define the dependence of the function $u$ on the parameters $\left(s_{2k+1}\right)_{1\leq k\leq g-1}$:
            \begin{align}
                0 &= R_{2g+1}[u] + \sum_{\ell = 0}^{g-1} \frac{2\ell+1}{2}s_{2\ell+1}R_{2\ell-1}[u], \label{string-eq}\\
                \frac{\partial u}{\partial s_{2\ell+1}} &= 2\frac{\partial}{\partial x} R_{2\ell+1}[u], \qquad \forall\, \ell \in \llbracket 1,g-1\rrbracket \label{KdV-flows}
            \end{align}
            where we have defined $s_1:=x$ for convenience. For practical reasons we also write the vector $\mathbf{s}:=\left(s_1:=x,s_3,s_5,\cdots s_{2g-1}\right)\in \CC^{g}$. We also define:
            \begin{equation}\label{Defcn}
            c_{2\ell-1}({\bf s}) := 
            \begin{cases}
                1, & \text{if }\, \ell = g+1,\\
                0, & \text{if }\,\ell = g,\\
                \frac{2\ell+1}{2}s_{2\ell+1}, & \forall\, \ell \in \llbracket 0, g-1\rrbracket.
            \end{cases}
        \end{equation}
        Note in particular that $c_{-1}(\mathbf{s})=\frac{s_1}{2}=\frac{x}{2}$. Equation \eqref{string-eq} is usually referred to as the ``string equation", while \eqref{KdV-flows} is usually referred to as the KdV flows.
    \end{definition}

\begin{remark}
Compatibility means that the KdV flows commute and preserve the
normalized string constraint, so that admissible local initial data
extend consistently in the times. The existence of one common solution
alone would not establish this property. The KdV hierarchy contains the
flows \eqref{KdV-flows} for all $\ell\geq0$; the string equation imposes
the additional reduction defining the $g^{\rm th}$ PI member.
The terminology originates in two-dimensional quantum gravity and
random matrix models \cite{DS}. In the differential-operator
formulation of \autoref{Appendix-A}, it is expressed as
$[Q,P^{(g)}]=1$, where $Q$ has order $2$ and $P^{(g)}$ has order
$2g+1$, together with the integration-constant normalization discussed
below.
\end{remark}

\begin{remark}In the reduction of the KdV hierarchy providing the $(2,2g+1)$ minimal models, the times are usually denoted $\left(t_{2k+1}\right)_{1\leq k\leq g-1}$. To avoid confusion with the irregular times of the isomonodromic approach, we use $\textit{s}$ instead of $\textit{t}$.
\end{remark}

\begin{remark}
The normalization $c_{2g-1}=0$ amounts to fixing $s_{2g+1}=0$.
Before this normalization, the corresponding parameter is redundant
under the usual affine normalization of the spectral problem
\cite{Takasaki}. This does not mean that its unreduced KdV vector field
vanishes. It is the time $s_{2g+1}$, not $s_{2g-1}$, that is fixed.
This convention agrees with \autoref{AssumptionTrivial} under the
identification in \autoref{TheoIdentificationTimes}.
\end{remark}

Note that from \eqref{string-eq} and \eqref{Defcn} the string equation reads:
\begin{equation}\label{RelationRR}
    0= \sum_{\ell = 0}^{g+1} c_{2\ell-1}(\mathbf{s}) R_{2\ell-1}[u]=R_{2g+1}[u]+\sum_{\ell=1}^{g-1}c_{2\ell-1}(\mathbf{s})R_{2\ell-1}[u] +\frac{x}{2}.
\end{equation}

Following \autoref{DefP1HierarchyReductionKP}, the first examples of the PI hierarchy from the reduction of the KP hierarchy are:
\begin{itemize}
    \item $g=1$, i.e. the first element  of the PI hierarchy. It is given by 
    \beq 0=R_{3}[u]+\frac{1}{2}xR_{-1}[u]=\frac{1}{8}u_{xx}+\frac{3}{8}u^2+\frac{1}{2}x\eeq
    and one recovers the standard\footnote{In the literature, there are many normalizations of the PI equation: another common one (cf. \cite{JimboMiwa,DK0}) is $U_{XX} = 6U^2 + X$. One can obtain this equation from \eqref{P1Recovered} by the change of variables $x= 2^{-1/5}X$, $u = -2^{7/5}U$.} PI equation satisfied by $u(x)$:
    \beq\label{P1Recovered}  0 =\frac{1}{4}u_{xx}+\frac{3}{4}u^2+x\eeq
    \item $g=2$, i.e. the second element of the PI hierarchy \footnote{To obtain the normalization that appears in \cite{Dubrovin-06}, for instance, one must make the following change of variables in equation \eqref{SecondMemberxs3}: $x=30^{1/7}X$, $s_3 = -\frac{1}{3}(30)^{3/7}T$, $u = 2 (30)^{-2/7} U$.}. It is given from \autoref{DefP1HierarchyReductionKP} as a compatible set of PDEs for $u(x,s_3)$:
\bea 0&=&R_5[u]+ \frac{3}{2}s_3R_{1}[u]+\frac{1}{2}xR_{-1}[u]\cr
\frac{\partial}{\partial s_3}u&=&2\frac{\partial}{\partial x} R_3[u]
\eea
i.e. from \eqref{FirstRk}:
\bea \label{SecondMemberxs3}0&=&\frac{1}{16}u_{xxxx}+\frac{5}{8}uu_{xx}+\frac{5}{16}(u_x)^2+\frac{5}{8}u^3+ \frac{3}{2}s_3u+x \cr
\frac{\partial}{\partial s_3}u&=&
\frac{1}{4}u_{xxx} + \frac{3}{2}uu_x
\eea
For fixed $s_3$, the first equation of \eqref{SecondMemberxs3} is already
a fourth-order ODE in $x$, namely the second member of the PI hierarchy.
Differentiating it with respect to $s_3$ and substituting the KdV flow
gives the following necessary differential consequence:
\bea 0&=& \frac{1}{64}u^{(7)}+\frac{3}{32}\left(10u^{(3)}u^{(2)}+5u_xu^{(4)}+u u^{(5)}\right)+\frac{5}{8}u^{(2)}\left(\frac{1}{4}u^{(3)}+\frac{3}{2}u u^{(1)}\right)\cr&&
+\frac{5}{8}u\left(\frac{1}{4}u^{(5)}+\frac{3}{2}u u^{(3)}+\frac{9}{2}u^{(1)} u^{(2)} \right)
+\frac{5}{8}u^{(1)}\left(\frac{1}{4}u^{(4)}+\frac{3}{2}(u^{(1)})^2+ \frac{3}{2}u u^{(2)}\right)\cr&&
+\frac{15}{8}u^2\left(\frac{1}{4}u^{(3)}+\frac{3}{2}u u^{(1)}\right)+\frac{3}{2}u+\frac{3}{2}s_3\left(\frac{1}{4}u^{(3)}+\frac{3}{2}u u^{(1)}\right) 
\eea
where we have denoted $u^{(k)}=\partial^k u$ the $k^{\text{th}}$ derivative of $u$ with respect to $x$ for any $k\geq 1$. Multiplying by $u$, one can then use the first equation of \eqref{SecondMemberxs3} to remove the dependence in $s_3$:
\bea \label{ODESecondMemberP1}0&=& \frac{1}{64}uu^{(7)}+\frac{3}{32}u\left(10u^{(3)}u^{(2)}+5u_xu^{(4)}+u u^{(5)}\right)+\frac{5}{8}uu^{(2)}\left(\frac{1}{4}u^{(3)}+\frac{3}{2}u u^{(1)}\right)\cr&&
+\frac{5}{8}u^2\left(\frac{1}{4}u^{(5)}+\frac{3}{2}u u^{(3)}+\frac{9}{2}u^{(1)} u^{(2)} \right)
+\frac{5}{8}uu^{(1)}\left(\frac{1}{4}u^{(4)}+\frac{3}{2}(u^{(1)})^2+ \frac{3}{2}u u^{(2)}\right)\cr&&
+\frac{15}{8}u^3\left(\frac{1}{4}u^{(3)}+\frac{3}{2}u u^{(1)}\right)+\frac{3}{2}u^2\cr&&
-\left(\frac{1}{16}u^{(4)}+\frac{5}{8}uu^{(2)}+\frac{5}{16}(u^{(1)})^2+\frac{5}{8}u^3+x\right)\left(\frac{1}{4}u^{(3)}+\frac{3}{2}u u^{(1)}\right) 
\eea
Equation \eqref{ODESecondMemberP1} is a necessary consequence of
\eqref{SecondMemberxs3}; it is not asserted to be an equivalent
definition of the second PI member.
\end{itemize}

\subsection{Definition from compatibility of linear differential equations}

\medskip

We can alternatively define the PI hierarchy as the compatibility of a linear system of differential equations on a matrix wave function $\Psi(\lambda;{\bf s})$. To see this, we first define a collection of formal Laurent series in a new ``spectral variable" $\lambda$, whose coefficients are differential polynomials in the function $u$:

\begin{definition}[Generating matrix and auxiliary matrices]\label{DefWaveMatrixFunction}Using \eqref{Lenard-recursion}, we define the following formal Laurent series in $\lambda$:
\begin{align}
        B(\lambda) &:= \sum_{k=0}^{\infty}R_{2k-1}[u]\lambda^{-k} \overset{\lambda\to \infty}{=} 1 + \frac{R_1[u]}{\lambda} + \frac{R_3[u]}{\lambda^2} + \cdots,\\
        A(\lambda) &:= -\frac{1}{2}\frac{\partial}{\partial x} B(\lambda) \overset{\lambda\to \infty}{=} O\left(\lambda^{-1}\right),\\
        C(\lambda) &:= -\frac{1}{2}\frac{\partial^2}{\partial x^2}B(\lambda) +(\lambda-u)B(\lambda) \overset{\lambda\to \infty}{=} \lambda -\frac{1}{2}u + O\left(\lambda^{-1}\right).
    \end{align}
We then regroup them to define the $2\times 2$ matrix
    \begin{equation}\label{DefU}
        U(\lambda) := 
        \begin{pmatrix}
            A(\lambda) & B(\lambda)\\
            C(\lambda) & -A(\lambda)
        \end{pmatrix},
    \end{equation}
Finally, for each $n\in \mathbb{N}$, we define the following matrix:
    \begin{equation}\label{Un-def}
        \forall \, n\geq 0\,:\, \mathcal{U}_{2n+1}(\lambda) := \left[\lambda^{n} U(\lambda)\right]_{\infty,+} -
        \begin{pmatrix}
            0 & 0\\
            R_{2n+1}[u] & 0
        \end{pmatrix}.
    \end{equation}
\end{definition}

The first examples of \autoref{DefWaveMatrixFunction} are
\begin{itemize}
    \item $n=0$: \begin{equation*}
        \mathcal{U}_{1}(\lambda) = 
        \begin{pmatrix}
            0 & 1\\
            \lambda- u & 0
        \end{pmatrix}.
    \end{equation*}
    \item $n=1$:  \begin{equation*}
        \mathcal{U}_{3}(\lambda) =  
        \begin{pmatrix}
            -\frac{1}{4}u_x & \lambda + \frac{1}{2}u\\
            \lambda^2-\frac{1}{2}u\lambda-\frac{1}{2}u^2-\frac{1}{4}u_{xx} & \frac{1}{4}u_x
        \end{pmatrix}.
    \end{equation*}
    \item $n=2$:\footnotesize{\begin{equation*}
        \mathcal{U}_{5}(\lambda) =  
        \begin{pmatrix}
            -\frac{1}{4}u_x\lambda - \frac{3}{8}uu_x-\frac{1}{16}u_{xxx} & \lambda^2 + \frac{1}{2}u\lambda + \frac{3}{8}u^2+\frac{1}{8}u_{xx}\\
            \lambda^3 - \frac{1}{2}u\lambda^2 - \frac{1}{8}\left(u^2+u_{xx}\right)\lambda -\frac{1}{16}u_{xxxx}-\frac{1}{2}uu_{xx}-\frac{3}{8}(u_x)^2-\frac{3}{8}u^3 & \frac{1}{4}u_x\lambda + \frac{3}{8}uu_x+\frac{1}{16}u_{xxx}
        \end{pmatrix}
    \end{equation*}}\normalsize{}
\end{itemize}

\bigskip

The linear compatibility equations recover the derivative of the
string equation. To recover the normalization in \eqref{string-eq},
define
\begin{equation}\label{string-normalization}
 F_g[u;\mathbf s]
 :=R_{2g+1}[u]
   +\sum_{\ell=1}^{g-1}\frac{2\ell+1}{2}s_{2\ell+1}R_{2\ell-1}[u]
   +\frac{x}{2}.
\end{equation}

\begin{proposition}\label{PI-matrix-prop}
The $g^{\rm th}$ PI member \eqref{string-eq}, \eqref{KdV-flows} is
equivalent to compatibility of the following linear system, supplemented
by the normalization $F_g[u;\mathbf s]=0$ at one value of $x$ for each
fixed choice of $(s_3,\ldots,s_{2g-1})$:
        \bea\label{DefLaxSystem}
        \frac{\partial \Psi}{\partial x} &=& \mathcal{U}_1(\lambda)\Psi,\cr
        \frac{\partial \Psi}{\partial \lambda} &=& \mathcal{A}^{(g)}(\lambda)\Psi := \left(\mathcal{U}_{2g+1}(\lambda) + \sum_{\ell=1}^{g-1} \frac{2\ell+1}{2}s_{2\ell+1}\mathcal{U}_{2\ell-1}(\lambda)\right)\Psi,\cr
        &=& \left(\sum_{\ell=1}^{g+1} c_{2\ell-1}(\mathbf{s})\,\mathcal{U}_{2\ell-1}(\lambda)\right)\Psi,\cr
        \frac{\partial \Psi}{\partial s_{2\ell+1}} &=& \mathcal{U}_{2\ell+1}(\lambda)\Psi, \qquad \forall \, \ell\in \llbracket1,g-1\rrbracket.
    \eea
\end{proposition}

Compatibility means that the differential operators of this system
commute. The matrix zero-curvature equations are given by
\autoref{lemma-A}, \autoref{lemma-B} and \autoref{lemma-C}.
The first lemma gives the KdV flows; the second gives
$\partial_xF_g=0$. The stated normalization then gives $F_g=0$,
and the third lemma supplies the remaining compatibility equations.
Conversely, the normalized string equation and the KdV flows imply
all these zero-curvature equations. The underlying computations are
recalled in \autoref{Appendix-A}.

\begin{lemma} \label{lemma-A}
    The compatibility conditions 
        \begin{equation}\label{Eq423}
            0 = \frac{\partial \mathcal{U}_{2\ell+1}}{\partial x} - \frac{\partial \mathcal{U}_{1}}{\partial s_{2\ell+1}} + \left[\mathcal{U}_{2\ell+1},\mathcal{U}_{1}\right]
        \end{equation}
    are equivalent to \eqref{KdV-flows}, $\forall \, \ell\in \llbracket1,g-1\rrbracket$.
\end{lemma}
    \begin{proof}
        See \autoref{KdV-flows-prop}.
    \end{proof}
\begin{lemma}\label{lemma-B}
Assume \eqref{Eq423} for $\ell\in \llbracket 1,g-1\rrbracket$. The condition
\[
 \partial_x\mathcal A^{(g)}
 -\partial_\lambda\mathcal U_1
 +[\mathcal A^{(g)},\mathcal U_1]=0
\]
is equivalent to $\partial_xF_g=0$, with $F_g$ defined in
\eqref{string-normalization}. It is equivalent to \eqref{string-eq}
after imposing $F_g=0$ at one value of $x$ for each fixed choice
of the remaining times.
\end{lemma}
\begin{proof}
The definitions of the matrices and the Lenard recursion give
\[
 \partial_x\mathcal U_{2n+1}
       +[\mathcal U_{2n+1},\mathcal U_1]
       =-2\partial_xR_{2n+1}\,E_{21},
 \qquad
 \partial_\lambda\mathcal U_1=E_{21}.
\]
Consequently,
\[
 \partial_x\mathcal A^{(g)}
 -\partial_\lambda\mathcal U_1
 +[\mathcal A^{(g)},\mathcal U_1]
       =-2\partial_xF_g\,E_{21}.
\]
This proves the assertion, including the need to fix the integration
constant.
\end{proof}
    
\begin{lemma}\label{lemma-C}
    Given the compatibility conditions 
        \begin{align*}
            0&=\frac{\partial \mathcal{A}^{(g)}}{\partial x} - \frac{\partial \mathcal{U}_{1}}{\partial \lambda} + \left[\mathcal{A}^{(g)},\mathcal{U}_{1}\right],\\
            0&= \frac{\partial \mathcal{U}_{2\ell+1}}{\partial x} - \frac{\partial \mathcal{U}_{1}}{\partial s_{2\ell+1}} + \left[\mathcal{U}_{2\ell+1},\mathcal{U}_{1}\right],\qquad \forall\, \ell\in \llbracket1,g-1\rrbracket,
        \end{align*}
    the compatibility conditions 
        \begin{align*}
            0&=\frac{\partial \mathcal{A}^{(g)}}{\partial s_{2\ell+1}} - \frac{\partial \mathcal{U}_{2\ell+1}}{\partial \lambda} + \left[\mathcal{A}^{(g)},\mathcal{U}_{2\ell+1}\right], \qquad \forall\, \ell\in \llbracket 1,g-1\rrbracket,\\
            0 &= \frac{\partial \mathcal{U}_{2\ell+1}}{\partial s_{2m+1}} - \frac{\partial \mathcal{U}_{2m+1}}{\partial s_{2\ell+1}} + \left[\mathcal{U}_{2\ell+1},\mathcal{U}_{2m+1}\right],\qquad \forall\, (\ell,m)\in\llbracket 1,g-1\rrbracket^2
        \end{align*}
    are satisfied. 
\end{lemma}
\begin{proof}
    See \autoref{Lemma3-prop}.
\end{proof}

\subsection{Eigenvalues and spectral invariants}

The eigenvalues of $\mathcal A^{(g)}(\lambda)$ are algebraic spectral
quantities. They must be distinguished from the diagonal entries of
the formally transformed differential system in \autoref{PropDiago}.
Indeed,
\[
 \hat L=G_\infty^{-1}L_\infty G_\infty
              -G_\infty^{-1}\partial_\lambda G_\infty,
 \qquad
 \det\hat L
 =\det\left(L_\infty-(\partial_\lambda G_\infty)G_\infty^{-1}\right).
\]
The derivative term is essential; in particular the common logarithmic
term of the formal wave normalization cannot be inserted into the
algebraic eigenvalues of a traceless matrix.

\begin{definition}[Eigenvalues and spectral polynomial]\label{Defxih} We shall denote $\left(\xi_{g,1}(\lambda),\xi_{g,2}(\lambda)\right)$ the eigenvalues of $\mathcal{A}^{(g)}(\lambda)$ understood as formal Laurent series in $\lambda^{1/2}$. We also define
\beq h^{(g)}(\lambda):=-\det \mathcal{A}^{(g)}(\lambda)\eeq
understood as a formal Laurent series at $\lambda=\infty$.   
\end{definition}

\medskip

Since $\mathcal A^{(g)}$ is traceless, its eigenvalues are opposite.
We order them by their leading terms and write
\begin{align}\label{Formxi}
 \xi_{g,1}(\lambda)
 &=\sqrt{h^{(g)}(\lambda)}
   =\sum_{k=-\infty}^{g+1}\alpha_{2k-1}\lambda^{k-\frac12},
 \qquad \alpha_{2g+1}=1,\nonumber\\
 \xi_{g,2}(\lambda)&=-\xi_{g,1}(\lambda).
\end{align}
Their first coefficients are determined by the times as follows.

    \begin{proposition}\label{PropositionAsymptoticExpansionP1}
        Fix $g\geq 1$, and let $(\xi_{g,1}(\lambda),\xi_{g,2}(\lambda))$ denote the eigenvalues of the matrix $\mathcal{A}^{(g)}(\lambda)$. Then, as $\lambda\to \infty$,
            \bea\label{xiexp}
                \xi_{g,1}(\lambda) &=& \lambda^{g+\frac{1}{2}} +\sum_{\ell=1}^{g}  
                c_{2\ell-1}(\mathbf{s})\lambda^{\ell-\frac{1}{2}} +\frac{x}{2}\lambda^{-\frac{1}{2}} +O\left(\lambda^{-\frac{3}{2}}\right)
                \cr
                \xi_{g,2}(\lambda) &=& -\lambda^{g+\frac{1}{2}} -\sum_{\ell=1}^{g} 
                c_{2\ell-1}(\mathbf{s})\lambda^{\ell-\frac{1}{2}} - \frac{x}{2}\lambda^{-\frac{1}{2}} +O\left(\lambda^{-\frac{3}{2}}\right).
           \eea 
    \end{proposition}
\begin{proof}
Put
\[
 f(\lambda)=\sum_{k=1}^{g+1}c_{2k-1}\lambda^{k-1},
 \qquad
 S=\sum_{k=1}^{g+1}c_{2k-1}R_{2k-1}=-\frac{x}{2}.
\]
The last identity is the normalized string equation. From
\eqref{Un-def},
\begin{equation}\label{Rewritin}
 \mathcal A^{(g)}=[fU]_{\infty,+}-S E_{21}.
\end{equation}
The generating identity in \autoref{DefinitionRks} gives
$A(\lambda)^2+B(\lambda)C(\lambda)=\lambda$, hence
$\det U(\lambda)=-\lambda$. The coefficient of $\lambda^{-1}$
in $fB$ is $S$. Therefore
\begin{align}\label{eqCoherence}
 \mathcal A^{(g)}_{11}&=fA+O(\lambda^{-1}),&
 \mathcal A^{(g)}_{12}&=fB-S\lambda^{-1}+O(\lambda^{-2}),\nonumber\\
 \mathcal A^{(g)}_{21}&=fC-S+O(\lambda^{-1}).
\end{align}
Using $fB=\lambda^g+O(\lambda^{g-1})$ and
$fC=\lambda^{g+1}+O(\lambda^g)$ gives
\begin{equation}\label{exph1}
 h^{(g)}(\lambda)
   =\lambda f(\lambda)^2-2S\lambda^g+O(\lambda^{g-1})
   =\lambda f(\lambda)^2+x\lambda^g+O(\lambda^{g-1}).
\end{equation}
Let
\[
 F(\lambda)=\sum_{k=0}^{g+1}c_{2k-1}\lambda^{k-\frac12}
           =\lambda^{\frac{1}{2}}f(\lambda)+\frac{x}{2}\lambda^{-\frac{1}{2}}.
\]
Since $F^2=\lambda f^2+xf+\frac{x^2}{4}\lambda^{-1}$ and
$f=\lambda^g+O(\lambda^{g-1})$, we obtain
\begin{equation}\label{exph}
 h^{(g)}-F^2=O(\lambda^{g-1}).
\end{equation}
As $\xi_{g,1}+F=2\lambda^{g+1/2}(1+O(\lambda^{-1}))$,
\[
 \xi_{g,1}-F
   =\frac{h^{(g)}-F^2}{\xi_{g,1}+F}=O\left(\lambda^{-\frac{3}{2}}\right).
\]
Thus
\begin{equation}\label{SolIDentification}
 \alpha_{2k-1}=c_{2k-1},\qquad 0\leq k\leq g+1,
\end{equation}
with $c_{2g+1}=1$ and $c_{-1}=\frac{x}{2}$, proving the proposition.
\end{proof}

The next step is to refine the formal expansion at infinity of $\left(\xi_{g,1}(\lambda),\xi_{g,2}(\lambda)\right)$ and $h^{(g)}(\lambda)$. We introduce the following definition:

\begin{definition}[Spectral Hamiltonians]\label{DefSpectralHamiltonian} We define $(H_k)_{1\leq k\leq g}$ by
\bea \xi_{g,1}(\lambda)&=&\sum_{k=0}^{g+1}c_{2k-1}\lambda^{k-\frac{1}{2}} +\sum_{k=1}^g H_k \lambda^{-k-\frac{1}{2}}+O\left(\lambda^{-g-\frac{3}{2}}\right)\cr
   \xi_{g,2}(\lambda)&=&-\sum_{k=0}^{g+1}c_{2k-1}\lambda^{k-\frac{1}{2}} -\sum_{k=1}^g H_k \lambda^{-k-\frac{1}{2}}+O\left(\lambda^{-g-\frac{3}{2}}\right) 
   \eea
   They form a set $\mathbf{H}:=(H_k)_{1\leq k\leq g}$ referred to as the set of ``spectral Hamiltonians".
\end{definition}

By definition, $H_k=\alpha_{-2k-1}$. We use the term
\emph{spectral Hamiltonians} for these coefficients of the eigenvalue
expansion. With this normalization, the isospectral generators below
are $2H_k$, not $H_k$; see \autoref{PropSpectralHamInvariants} and
\autoref{Appendix-B}. The Hamiltonians of the nonautonomous evolutions
in the spectral Darboux coordinates also require a correction.
Time-dependent canonical coordinates adapted to isospectral
deformations are discussed in
\cite{MarchalAlameddineIsospectralIsomono2023,MohamadP1Isospectral}.

It is obvious that the spectral Hamiltonians can be used to obtain a refinement of the expansion of the determinant of $\mathcal{A}^{(g)}(\lambda)$ at infinity. More precisely, one can use \autoref{DefSpectralHamiltonian} to obtain the next terms of the expansion of $h^{(g)}(\lambda)$.

\begin{definition}[Spectral invariants]\label{DefCasimirsSpectralInvariants}
Let
\[
 F(\lambda)=\sum_{m=0}^{g+1}c_{2m-1}\lambda^{m-\frac12}.
\]
We define
\begin{align}\label{DefI0}
 I_0(\lambda)&=[F(\lambda)^2]_{\infty,+},\nonumber\\
 \widetilde I_0(\lambda)
   &=\sum_{j=g}^{2g+1}
         [I_0(\lambda)]_{\lambda^j}\lambda^j,\nonumber\\
 I_0(\lambda)-\widetilde I_0(\lambda)
   &=\sum_{k=1}^g A_k(\mathbf s)\lambda^{g-k},
 \qquad
 A_k(\mathbf s)=
   \sum_{m=0}^{g+1-k}c_{2m-1}c_{2g-2k-2m+1}.
\end{align}
The spectral invariants $I_1,\ldots,I_g$ are determined by
\[
 h^{(g)}(\lambda)=I_0(\lambda)+\sum_{k=1}^g I_k\lambda^{g-k}.
\]
We use column vectors
$\mathbf I=(I_1,\ldots,I_g)^t$,
$\mathbf H=(H_1,\ldots,H_g)^t$ and
$\mathbf A=(A_1,\ldots,A_g)^t$. In particular, the coefficients
of $h^{(g)}-\widetilde I_0$ are $I_k+A_k$, not $I_k$.
The coefficients of $\widetilde I_0$ are Casimirs of the Poisson
bracket in \autoref{Appendix-B}.
\end{definition}

\begin{proposition}[Relation between spectral Hamiltonians and spectral invariants]\label{PropSpectralHamInvariants}
For $1\leq k\leq g$,
\[
 I_k=2\sum_{m=g+2-k}^{g+1}c_{2m-1}H_{m+k-g-1}.
\]
Equivalently,
\begin{equation}\label{EqMatrixFromI}
 \mathbf I=2C(\mathbf s)\mathbf H,
 \qquad
 C_{ij}(\mathbf s)=
 \begin{cases}
 c_{2g+1-2(i-j)}(\mathbf s),&1\leq j\leq i\leq g,\\
 0,&1\leq i<j\leq g.
 \end{cases}
\end{equation}
The matrix $C$ is lower triangular Toeplitz, with diagonal entries
$c_{2g+1}=1$.
\end{proposition}
\begin{proof}
Write
\[
 \xi_{g,1}=F+\sum_{j=1}^gH_j\lambda^{-j-\frac12}
                         +O(\lambda^{-g-\frac32}).
\]
Squaring and taking the polynomial part gives
\[
 h^{(g)}=I_0+
 2\left[F\sum_{j=1}^gH_j\lambda^{-j-\frac12}\right]_{\infty,+}.
\]
The square of the negative tail and its omitted terms have no
polynomial part. Extracting the coefficient of $\lambda^{g-k}$
proves the formula. The factor $2$ comes from the two cross terms.
\end{proof}

\subsection{Spectral Darboux coordinates}\label{SectionDarboux}
In the previous subsection, we defined $g$ times $\left(s_{2k+1}\right)_{0\leq k\leq g-1}$ that give the flows of the PI hierarchy. Since one has an underlying symplectic structure of dimension $2g$, one defines a set of Darboux coordinates whose evolutions with respect to these times characterize this symplectic structure. There are several choices for these coordinates. The most standard one is to take the so-called ``spectral Darboux coordinates" (See Section $8$ of \cite{Takasaki}).

\begin{definition}[Spectral Darboux coordinates (See, Section $8$ of \cite{Takasaki})]\label{DefinitionSpectralCoord}
On the open subset where $\mathcal A^{(g)}_{12}$ has $g$ distinct
roots, choose a local ordering and define
$(\lambda_j,\mu_j)_{j=1}^g$ by
\bea \left[\mathcal{A}^{(g)}(\lambda_j)\right]_{1,2}&=&0\,\,,\,\forall\, j\in \llbracket 1,g\rrbracket,\cr
\left[\mathcal{A}^{(g)}(\lambda_j)\right]_{1,1}&=&\mu_j\,\,,\,\forall\, j\in \llbracket 1,g\rrbracket.
\eea
In other words, $\boldsymbol{\lambda}:=\left(\lambda_j\right)_{1\leq j\leq g}$ are the zeros of $\left[\mathcal{A}^{(g)}(\lambda)\right]_{1,2}$ (which is a monic polynomial of degree $g$) and $\boldsymbol{\mu}:=\left(\mu_j\right)_{1\leq j\leq g}$ are the evaluations of $\left[\mathcal{A}^{(g)}(\lambda)\right]_{1,1}$ at these values.   
\end{definition}

It is well-known that these coordinates $(\boldsymbol{\lambda}, \boldsymbol{\mu})$ are canonical (See for example Section 8 of \cite{Takasaki} and references therein) with respect to the Poisson bracket associated to the so-called ``Mumford system''. We give a brief recounting of the properties of this Poisson bracket and a sketch of the proof that these coordinates are indeed canonical with respect to it in Appendix \ref{Appendix-B}. In particular, we have that
\beq\label{canonical-relations}\forall\, (j,k)\in \llbracket1,g\rrbracket^2\,:\,  \left\{\lambda_j,\lambda_k\right\}=0\, \,,\,\, \left\{\mu_j,\mu_k\right\}=0 \,\,,\,\, \left\{\lambda_j,\mu_k\right\}=\delta_{j,k},\eeq
which follows from \autoref{Canonical-proof},
and also
\beq\label{Casimirs} \forall \,j\in \llbracket1,g\rrbracket\,:\, \left\{\lambda_j, \td{I}_0(\lambda)\right\}=0=\left\{\mu_j,\td{I}_0(\lambda)\right\},\eeq
with $\td{I}_0(\lambda)$ defined in \autoref{DefCasimirsSpectralInvariants} (this is an immediate corollary of Proposition \ref{Poisson-Ik-prop}). Here the Poisson bracket $\{ .,.\}$ is the Poisson bracket associated to the $3g+1$-dimensional moduli space of the matrix $\mathcal{A}^{(g)}(\lambda)$ whose definition is detailed in \autoref{Appendix-B}. In the end, we have the following results

\begin{proposition}[Darboux coordinates property and Casimirs]\label{DarbouxCoordinatesProperty} The coordinates $(\boldsymbol{\lambda},\boldsymbol{\mu})$ are Darboux coordinates for the Poisson bracket. Moreover, coefficients of $\td{I}_0(\lambda)$ are Casimirs.
\end{proposition}

\begin{proof}The proof is immediate from \eqref{canonical-relations} and \eqref{Casimirs}.
    
\end{proof}

\subsection{Hamiltonian evolutions}
The next step is to determine the Hamiltonian evolutions of the spectral Darboux coordinates. In the $(2,2g+1)$ minimal model, following the notation of \cite{Takasaki}, we denote the Hamiltonians by $\left(K_{2\ell-1}\right)_{1\leq \ell\leq g}$. This corresponds to the following evolutions for the Darboux coordinates $\left(\boldsymbol{\lambda},\boldsymbol{\mu}\right):=\left(\lambda_j,\mu_j\right)_{1\leq j\leq g}$:
\bea \label{HamiltonianDefK} \forall \, \ell \in \llbracket 1,g\rrbracket\,:\, \frac{\partial \lambda_j}{\partial s_{2\ell-1}}  &=&\left\{\lambda_j,K_{2\ell-1}\right\}=\frac{\partial}{\partial \mu_j} K_{2\ell-1}(\boldsymbol{\lambda},\boldsymbol{\mu},\mathbf{s}),\cr
\frac{\partial \mu_j }{\partial s_{2\ell-1}} &=&\left\{\mu_j,K_{2\ell-1}\right\}= -\frac{\partial}{\partial \lambda_j} K_{2\ell-1}(\boldsymbol{\lambda},\boldsymbol{\mu},\mathbf{s}),
\eea
where the last equalities (i.e., the equality of the Poisson bracket of $K$ with respect to a canonical coordinate and the partial derivative with respect to that coordinate) follow from the fact that the coordinates $\{\lambda_j,\mu_j\}$ are canonical as proved in \eqref{canonical-relations}.

\begin{remark}
On a connected symplectic leaf, compatibility of the Hamiltonian
vector fields implies that
\begin{equation}\label{AddProp1}
 \mathcal F_{\ell m}:=
 \frac{\partial K_{2\ell-1}}{\partial s_{2m-1}}
 -\frac{\partial K_{2m-1}}{\partial s_{2\ell-1}}
 +\{K_{2\ell-1},K_{2m-1}\}
\end{equation}
is independent of the phase-space coordinates. Partial derivatives
here hold those coordinates fixed. The Jacobi identity gives the
Bianchi identity, which reduces to closedness of this time-dependent
two-form because $\mathcal F_{\ell m}$ is scalar on each leaf.
Locally on a contractible time domain, the Poincar\'e lemma therefore
provides functions $f_{2\ell-1}(\mathbf s)$ such that
$\widehat K_{2\ell-1}=K_{2\ell-1}+f_{2\ell-1}$ satisfies
\begin{equation}\label{StrongerProp}
 \frac{\partial\widehat K_{2\ell-1}}{\partial s_{2m-1}}
 -\frac{\partial\widehat K_{2m-1}}{\partial s_{2\ell-1}}
 +\{\widehat K_{2\ell-1},\widehat K_{2m-1}\}=0.
\end{equation}
A global choice requires exactness of the corresponding two-form.
Adding these functions does not change the Hamiltonian vector fields.
Flatness alone, however, does not specify those vector fields.
For the comparison below we keep the explicit normalization of
\autoref{PropTakasakiTheorem3}; we do not impose an additional flat
normalization of the Hamiltonian one-form.
\end{remark}

\begin{remark}
We retain the notation
$\mathbf K=(K_1,K_3,\ldots,K_{2g-1})^t$ for the corrected Hamiltonians.
Care is required when comparing spectral-Hamiltonian conventions:
the coefficients $H_k$ of \autoref{DefSpectralHamiltonian} are
normalized by the eigenvalue expansion. The interpolation
Hamiltonians based on $\widetilde I_0$ are instead the components of
$2\mathbf H+C^{-1}\mathbf A$.
\end{remark}

Following the work of \cite{Takasaki}, we first define the following quantities.

\begin{definition}[Definition of $\left(\td{R}_n(\lambda)\right)_{n\geq 0}$]\label{DeftdRn}We define the following polynomials in $\lambda$:
  \bea\label{DefRlambda} \td{R}_0(\lambda)&:=&1\cr
\td{R}_n(\lambda)&:=&\lambda^n+R_1[u]\lambda^{n-1}+R_3[u]\lambda^{n-2}+\dots+ R_{2n-1}[u]=\left[\mathcal{U}_{2n+1}(\lambda)\right]_{1,2}\cr
&=& \sum_{k=0}^n R_{2k-1}[u]\lambda^{n-k}\,\,\,,\,\, \forall \, n\geq 1
\eea  
\end{definition}

\begin{definition}[Definition of $\left(\beta_n^{(g)}(\lambda)\right)_{0\leq n\leq g}$]\label{Defbetan}We define $\beta_0(\lambda):=1$, and
\beq\label{Defbeta} \forall\, g\geq 1\,: \beta_g(\lambda):=\left[\mathcal{A}^{(g)}(\lambda)\right]_{1,2}=\lambda^g +\sum_{k=0}^{g-1}\beta_{g,g-k} \lambda^k =\prod_{k=1}^g (\lambda-\lambda_k)\eeq
We then set
\beq
\beta_n^{(g)}(\lambda) := \left[\lambda^{n-g}\beta_g(\lambda)\right]_{\infty,+}, \qquad\qquad \forall\, n\in \llbracket 0,g\rrbracket.
\eeq
\end{definition}
Note that $\beta_g^{(g)}(\lambda)=[\beta_{g}(\lambda)]_{\infty,+}=\beta_g(\lambda)$. We have also:
\beq \beta_g'(\lambda)=\sum_{n=1}^g\prod_{i\neq n} (\lambda-\lambda_i) \,\, \Rightarrow \beta_g'(\lambda_j)=\prod_{i\neq j}(\lambda_j-\lambda_i) \,\,,\,\, \forall \, j\in \llbracket 1,g\rrbracket.\eeq
Moreover, we have the important relations:
    \begin{proposition}
The following relation between $\left(\td{R}_n(\lambda)\right)_{n\geq 0}$ and $\left(\beta_n^{(g)}(\lambda)\right)_{0\leq n\leq g}$ holds:
        \beq \label{eqbetabold2}\begin{pmatrix}
            \beta_0^{(g)}(\lambda)\\\beta_1^{(g)}(\lambda)\\ \vdots \\ \beta_{g}^{(g)}(\lambda)
        \end{pmatrix}= \td{C}(\mathbf{s})\begin{pmatrix} \td{R}_{0}(\lambda)\\ \vdots \\ \td{R}_{g}(\lambda)\end{pmatrix}
        \eeq
where the matrix $\td{C}(\mathbf{s})$ is the $(g+1)\times (g+1)$ lower triangular Toeplitz matrix
            \begin{equation}
                \td{C}(\mathbf{s}) = 
                    \begin{pmatrix}1&0&\dots&\dots& \dots&0\\
                        c_{2g-1}& 1&0& \ddots& \ddots &\vdots\\
                        c_{2g-3}&\ddots&\ddots& \ddots&\ddots&\vdots\\
                        \vdots& \ddots& \ddots&\ddots&\ddots&\vdots \\
                        c_{3}& c_{5} & \ddots & \ddots& 1 &0\\
                        c_1 & c_3 & \dots & \dots & c_{2g-1}& 1
                    \end{pmatrix}.
            \end{equation}
           As a subcase we also have:
           \beq \label{eqbetabold}\boldsymbol{\beta}(\lambda):=\begin{pmatrix}
            \beta_0^{(g)}(\lambda)\\\beta_1^{(g)}(\lambda)\\ \vdots \\ \beta_{g-1}^{(g)}(\lambda)
        \end{pmatrix}= C(\mathbf{s}) \td{\mathbf{R}}(\lambda):=C(\mathbf{s})\begin{pmatrix} \td{R}_{0}(\lambda)\\ \vdots \\ \td{R}_{g-1}(\lambda)\end{pmatrix} 
        \eeq
        where the matrix $C(\mathbf{s})$ is defined in \eqref{EqMatrixFromI}. 
        Finally we have
        \beq \label{betan}\beta_{g,n}=\sum_{j=0}^n c_{2g-2j +1}(\mathbf{s})R_{2n-2j-1}[u] \,\,,\,\,\forall\, n\in \llbracket 1, g\rrbracket\eeq
    \end{proposition}
    \begin{proof}
    By definition, we have that
    \begin{align*}
        \beta_g(\lambda) &= [\mathcal{A}^{(g)}(\lambda)]_{1,2} = \sum_{\ell=0}^g \left(c_{2\ell+1}({\bf s}) \sum_{j=0}^{\ell} R_{2\ell-2j-1}\lambda^j\right)=\sum_{j=0}^{g}\left(\sum_{\ell=j}^g c_{2\ell +1}(\mathbf{s})R_{2\ell -2j-1}\right) \lambda^j.
    \end{align*}
    This gives $\beta_{g,g-j}=\underset{\ell=j}{\overset{g}{\sum}} c_{2\ell +1}(\mathbf{s})R_{2\ell -2j-1}$ for all $j\in \llbracket 0, g-1\rrbracket$, i.e. 
    $\beta_{g,n}=\underset{j=0}{\overset{n}{\sum}} c_{2g-2j +1}(\mathbf{s})R_{2n-2j-1}$ for $n\in \llbracket 1,g\rrbracket$.
    Using the definition of $\beta^{(g)}_n(\lambda)$ and the fact that $\left[\lambda^{n-g}\tilde{R}_{\ell}(\lambda)\right]_{\infty,+} \equiv 0$ for $\ell < g-n$, we find that for all $n\in \llbracket 0,g\rrbracket$:
    \bea\label{esnew}
        \beta_n^{(g)}(\lambda) &=& \sum_{\ell=g-n}^{g}\left(c_{2\ell+1}({\bf s})\sum_{j=g-n}^{\ell}R_{2\ell-2j-1}\lambda^{j+n-g}\right)\cr
                               &=& \sum_{\ell=g-n}^{g}c_{2\ell+1}({\bf s})\left(\sum_{k=0}^{\ell+n-g}R_{2(\ell+n-g)-2k-1}\lambda^{k}\right)\cr
                               &=& \sum_{\ell=g-n}^{g}c_{2\ell+1}({\bf s})\tilde{R}_{\ell+n-g}(\lambda) = \sum_{k=0}^{n}c_{2g-2(n-k)+1}({\bf s})\tilde{R}_{k}(\lambda).
    \eea
    ending the proof.
    \end{proof}

\begin{remark}Note that the scalar equations \eqref{esnew} are equivalent to the matrix form \eqref{eqbetabold2}. The submatrix $C(\mathbf{s})$ of $\td{C}(\mathbf{s})$ is introduced because it shall appear in the identification of both formalisms because it also relates $\left(I_k\right)_{k=1}^g$ and $\left(H_k\right)_{k=1}^g$ as explained in \eqref{EqMatrixFromI}. 
\end{remark}

For the remainder of this section write $\beta=\beta_g$ and
$\beta_n^{(g)}=[\lambda^{n-g}\beta_g]_{\infty,+}$.
Since $h^{(g)}(\lambda_j)=\mu_j^2$, Lagrange interpolation gives
\[
 I_k=\sum_{j=1}^g
       \frac{\mu_j^2-I_0(\lambda_j)}{\beta_g'(\lambda_j)}
                         \beta_{k-1}^{(g)}(\lambda_j).
\]
Together with $\boldsymbol\beta=C\widetilde{\mathbf R}$ and
$\mathbf I=2C\mathbf H$, this yields
\[
 2H_k=\sum_{j=1}^g
       \frac{\mu_j^2-I_0(\lambda_j)}{\beta_g'(\lambda_j)}
                          \widetilde R_{k-1}(\lambda_j).
\]
Replacing $I_0$ by $\widetilde I_0$ adds $[C^{-1}\mathbf A]_k$.

\begin{proposition}[Hamiltonians in the spectral Darboux coordinates
(Theorem 3 of \cite{Takasaki})]\label{PropTakasakiTheorem3}
For $1\leq k\leq g$, the Hamiltonian for the $s_{2k-1}$ evolution is
\begin{align}
 K_{2k-1}
 &=\sum_{j=1}^g
    \frac{\mu_j^2-\widetilde I_0(\lambda_j)}
         {\beta_g'(\lambda_j)}\widetilde R_{k-1}(\lambda_j)
   -\sum_{j=1}^g
      \frac{\mu_j\widetilde R_{k-1}'(\lambda_j)}
           {\beta_g'(\lambda_j)}\nonumber\\
 &=2H_k+[C^{-1}\mathbf A]_k
   -\sum_{j=1}^g
      \frac{\mu_j\widetilde R_{k-1}'(\lambda_j)}
           {\prod_{i\ne j}(\lambda_j-\lambda_i)}.
\end{align}
Primes on $\beta_g$, $\beta_n^{(g)}$ and $\widetilde R_n$ denote
derivatives with respect to their spectral argument.
\end{proposition}

We use the convention
$\operatorname{Res}_{\lambda=\infty}f(\lambda)\,d\lambda
=-[f(\lambda)]_{\lambda^{-1}}$. When $d\lambda$ is suppressed below,
this same convention is understood.

\begin{lemma}[Expression of the correction term]\label{LemmaEquality}
For $1\leq k\leq g$,
\[
 \operatorname{Res}_{\lambda=\infty}
 \lambda^{-(g+1-k)}
 \sum_{j=1}^g\frac{\mu_j}{\lambda-\lambda_j}
          \prod_{m\ne j}\frac{\lambda-\lambda_m}{\lambda_j-\lambda_m}
 \,d\lambda
 =
 -\sum_{j=1}^g
   \frac{\mu_j(\beta_{k-1}^{(g)})'(\lambda_j)}
        {\beta_g'(\lambda_j)}.
\]
\end{lemma}
\begin{proof}
Write
$\beta_g(\lambda)=\underset{r=0}{\overset{g}{\sum}} b_r\lambda^{g-r}$, with $b_0=1$.
Then
\begin{equation}\label{betaprime}
 (\beta_{k-1}^{(g)})'(\lambda_j)
 =\sum_{m=1}^{k-1}m\,b_{k-1-m}\lambda_j^{m-1}
 =\left[\frac{\beta_g(\lambda)}
              {(\lambda-\lambda_j)^2}\right]_{\lambda^{g-k}}.
\end{equation}
The last equality follows by expanding
$(\lambda-\lambda_j)^{-2}
=\underset{m\geq 1}{\sum}m\lambda_j^{m-1}\lambda^{-m-1}$,
multiply by $\frac{\mu_j}{\beta_g'(\lambda_j)}$, and sum over $j$ using
the residue convention.
\end{proof}

\section{Correspondence between the two approaches}\label{sec5}

The purpose of this section is to identify the two formalisms introduced in \autoref{sec3} and \autoref{sec4} and to draw conclusions from this identification. This also provides a clear dictionary between both approaches.

\subsection{Identification of times}
We first recall that the equation $r_\infty-3 = g$ gives the correspondence between the order of the pole at infinity $r_{\infty}$ and the parameter $g$ on the minimal model side. Moreover, the first identification deals with the irregular times of the isomonodromic setting and the times introduced in the reduction of the KdV hierarchy. \autoref{PropositionAsymptoticExpansionP1} immediately allows to identify the times in both setups.

\begin{theorem}[Identification of the times]\label{TheoIdentificationTimes}
Under \autoref{AssumptionTrivial}, the time identification is
\[
 t_{\infty,2k-1}=(2k-1)s_{2k-1},
 \qquad 1\leq k\leq g,
 \qquad s_1=x,
\]
together with
\[
 t_{\infty,2g+1}=0,\qquad t_{\infty,2g+3}=2,
 \qquad t_{\infty,2j}=0\quad(1\leq j\leq g+2).
\]
For the vectors of nontrivial times
$\mathbf t=(t_{\infty,1},t_{\infty,3},\ldots,t_{\infty,2g-1})^t$
and $\mathbf s=(s_1,s_3,\ldots,s_{2g-1})^t$, put
$D=\operatorname{diag}(1,3,\ldots,2g-1)$. Then
\[
 \mathbf t=D\mathbf s,\qquad
 \partial_{s_{2k-1}}=(2k-1)\partial_{t_{\infty,2k-1}}.
\]
In expressions depending on all irregular times, the notation
$\mathbf t=D\mathbf s$ also includes the fixed times displayed above.
\end{theorem}
\begin{proof}
Compare \eqref{xiexp} with \eqref{LocalDiagP1Reduced} and use
$c_{2k-3}=(2k-1)s_{2k-1}/2$.
\end{proof}
\begin{remark}\label{RemarkS2gPlus1}
The normalization on the KdV side is
$s_{2g+1}=0$ and $s_{2g+3}=2/(2g+3)$, with higher times zero.
This extends the time dictionary to the fixed highest times.
The higher KdV vector fields can still be used as formal differential
derivations before restricting to this time slice.
\end{remark}

\subsection{Identification of Darboux coordinates and Lax matrices}

The identification of times does not by itself identify a point of the
phase space. We now construct that identification on the simple-root
chart, using the $2g$ initial values
$u,u_x,\ldots,u^{(2g-1)}$ of the normalized string equation.

\begin{theorem}[Identification of the Darboux coordinates and Lax matrices]\label{TheoIdentificationDarbouxCoordinates}
Fix the times according to \autoref{TheoIdentificationTimes}.
The map sending a jet
$(u,u_x,\ldots,u^{(2g-1)})$ to the normalized matrix
$\mathcal A^{(g)}$ is an invertible polynomial parametrization
of the normalized matrix phase space at these fixed times.
On the chart where its upper-right entry has distinct roots, it
identifies the local coordinates by
\[
 q_j=\lambda_j,\qquad p_j=\mu_j,\qquad 1\leq j\leq g,
\]
with the same ordering of roots. Under this map,
\[
 \hat L(\lambda)=\mathcal A^{(g)}(\lambda).
\]
Along the corresponding solutions, with a common normalization of
the wave matrix, the auxiliary matrices satisfy
\[
 (2k-1)\hat A_{\mathbf e_{2k-1}}(\lambda)
       =\mathcal U_{2k-1}(\lambda),\qquad 1\leq k\leq g.
\]
\end{theorem}
\begin{proof}
Write
\[
 \mathcal A^{(g)}
   =\begin{pmatrix}a&b\\d&-a\end{pmatrix},
 \qquad
 b(\lambda)=\lambda^g+\sum_{n=1}^g b_n\lambda^{g-n}.
\]
By the definitions of the auxiliary matrices,
\[
 a=-\frac12\partial_x b,\qquad
 d=(\lambda-u)b-\frac12\partial_x^2b.
\]
Here the highest derivative $u^{(2g)}$ is eliminated using the
normalized string equation, whose coefficient of that derivative
is $2^{-2g-1}$.

The Lenard recursion gives, for $n\geq1$,
\[
 R_{2n-1}=2^{1-2n}u^{(2n-2)}
        +\text{a differential polynomial of lower order}.
\]
Using \eqref{betan}, we therefore obtain
\[
 b_n=2^{1-2n}u^{(2n-2)}
             +\text{terms involving lower derivatives},
\]
and, on writing
$a(\lambda)=\sum_{n=1}^g a_n\lambda^{g-n}$,
\[
 a_n=-2^{-2n}u^{(2n-1)}
             +\text{terms involving lower derivatives}.
\]
Consequently the map from the jet to
$(b_1,a_1,\ldots,b_g,a_g)$ is triangular with nonzero constant
diagonal entries and has a polynomial inverse.

For completeness, a normalized matrix with the prescribed times is
uniquely determined by $a$ and $b$. Indeed,
\[
 -\det\hat L=\widetilde I_0(\lambda)
              +\text{a polynomial of degree at most }g-1,
 \qquad
 \hat L_{21}
 =\left[\frac{\widetilde I_0-a^2}{b}\right]_{\infty,+}.
\]
The prescribed high-degree determinant coefficients are the same
for $\mathcal A^{(g)}$, by
\autoref{PropositionAsymptoticExpansionP1}. This proves equality
of the lower-left entries as well. On the simple-root chart,
$b$ determines the ordered roots and $a$ is uniquely reconstructed
from its values $a(q_j)=p_j$ by Lagrange interpolation. Hence the
coordinate identification is invertible.

Finally, the compatible KdV/string system maps to the isomonodromic
system with these times. Choose the same wave matrix on the two
sides and use
$\partial_{s_{2k-1}}=(2k-1)\partial_{t_{\infty,2k-1}}$.
The auxiliary-matrix identity follows.
\end{proof}

\begin{corollary}
For the Poisson bracket of \autoref{Appendix-B}, transported through
the above matrix identification, the oper coordinates satisfy
\[
 \{q_i,q_j\}=\{p_i,p_j\}=0,\qquad \{q_i,p_j\}=\delta_{ij}.
\]
Thus the identification preserves the canonical symplectic forms:
\[
 \sum_{j=1}^g dq_j\wedge dp_j
       =\sum_{j=1}^g d\lambda_j\wedge d\mu_j.
\]
\end{corollary}
\begin{proof}
Apply \autoref{Canonical-proof} and the identities
$q_j=\lambda_j$, $p_j=\mu_j$.
\end{proof}

This corollary identifies the canonical Hamiltonian structures used in
the two coordinate descriptions. Equality with the intrinsic
symplectic form of the connection moduli space relies on the
geometric Darboux-coordinate identification used in
\cite{MarchalP1Hierarchy}; equality of the Lax matrices alone is
not an independent proof of that geometric identification.

\subsection{Identification of spectral invariants with $H_{\infty,k}$}
We may also identify the expansion of $\det \hat{L}(\lambda)$ at infinity using \eqref{ExpressionDethatL} and \autoref{DefCasimirsSpectralInvariants}. The most singular terms match and we get the following relation.

\begin{theorem}[Relation between $H_{\infty,k}$ and the spectral invariants $I_k$]\label{TheoHinftyI}
Define, for $1\leq k\leq g$,
\[
 T_k=
 \operatorname{Res}_{\lambda=\infty}
 \lambda^{-(g+1-k)}
 \sum_{i=1}^g\frac{p_i}{\lambda-q_i}
                  \prod_{j\ne i}\frac{\lambda-q_j}{q_i-q_j}\,d\lambda,
 \qquad
 \mathbf T=(T_1,\ldots,T_g)^t.
\]
Then
\begin{equation}\label{MatrixFromHinftyI}
 I_k=H_{\infty,g-k}-A_k-T_k,\qquad
 \mathbf H_\infty^{\rm rev}
       =\mathbf I+\mathbf A+\mathbf T,
\end{equation}
where
$\mathbf H_\infty^{\rm rev}
 =(H_{\infty,g-1},\ldots,H_{\infty,0})^t$
and $\mathbf A$ is defined in \eqref{DefI0}.
\end{theorem}
\begin{proof}
From \eqref{ExpressionDethatL}, with
\[
 F_q(\lambda)=
 \sum_{i=1}^g\frac{p_i}{\lambda-q_i}
                  \prod_{j\ne i}\frac{\lambda-q_j}{q_i-q_j},
\]
we have
\[
 h^{(g)}=-\det\hat L
   =-P_2+\sum_{j=0}^{g-1}H_{\infty,j}\lambda^j
       -\sum_{i=1}^g\frac{p_i}{\lambda-q_i}+F_q(\lambda).
\]
Since $-P_2=\widetilde I_0$ has no terms of degree below $g$,
coefficient extraction gives
\[
 [h^{(g)}]_{\lambda^{g-k}}
   =[-\det\hat L]_{\lambda^{g-k}}
   =H_{\infty,g-k}-T_k.
\]
The same coefficient is $I_k+A_k$, proving the result.
\end{proof}

\subsection{Identification of the Hamiltonians}
We may now use \autoref{Hamqp} and \autoref{PropSpectralHamInvariants} to upgrade the previous result to a relation between Hamiltonians and the spectral Hamiltonians:

\begin{theorem}[Relation between Hamiltonians and spectral Hamiltonians]\label{TheoHamH}
Under the time and coordinate identifications above, define
\[
 \mathcal H_k(\mathbf q,\mathbf p;\mathbf s)
  =(2k-1)\operatorname{Ham}^{(\mathbf e_{2k-1})}
                 (\mathbf q,\mathbf p;D\mathbf s),
 \qquad
 \boldsymbol{\mathcal H}=(\mathcal H_1,\ldots,\mathcal H_g)^t .
\]
Then
\begin{equation}\label{DefVectorR}
 \mathbf R:=\mathbf A+\mathbf T,\qquad
 \boldsymbol{\mathcal H}
    =2\mathbf H+C(\mathbf s)^{-1}\mathbf R.
\end{equation}
\end{theorem}
\begin{proof}
Let
$\mathbf H_\infty=(H_{\infty,0},\ldots,H_{\infty,g-1})^t$,
and let $J$ be the reversal matrix,
$J_{ij}=\delta_{i,g+1-j}$. The matrix $C$ is lower triangular
Toeplitz, so $J^2=I_g$ and $JC^tJ=C$.
Equation \eqref{RelationNuAlphaInfty} gives
\[
 C\,\boldsymbol\nu^{(\mathbf e_{2k-1})}
        =\frac1{2k-1}\mathbf e_{g+1-k},
\]
where on the right $\mathbf e_j$ denotes the $j^{\rm th}$ standard
basis vector of $\mathbb C^g$. Thus \eqref{DefHamReduced} implies
\[
 \boldsymbol{\mathcal H}
  =J(C^t)^{-1}\mathbf H_\infty
  =C^{-1}J\mathbf H_\infty
  =C^{-1}\mathbf H_\infty^{\rm rev}.
\]
Now use \eqref{MatrixFromHinftyI} and
$\mathbf I=2C\mathbf H$:
\[
 \boldsymbol{\mathcal H}
    =C^{-1}(\mathbf I+\mathbf A+\mathbf T)
    =2\mathbf H+C^{-1}(\mathbf A+\mathbf T).
\]
\end{proof}

The factor $2k-1$ is the chain-rule factor in the time change:
\[
 \operatorname{Ham}_{s_{2k-1}}
       (\boldsymbol\lambda,\boldsymbol\mu;\mathbf s)
 =(2k-1)\operatorname{Ham}^{(\mathbf e_{2k-1})}
       (\boldsymbol\lambda,\boldsymbol\mu;D\mathbf s).
\]
The term $[C^{-1}\mathbf A]_k$ depends only on the times.
The term $[C^{-1}\mathbf T]_k$ is the phase-dependent correction
needed in the spectral Darboux coordinates.

\bigskip

We now compare the explicitly normalized Hamiltonians in
\autoref{Hamqp} and \autoref{PropTakasakiTheorem3}.

\begin{theorem}[Identification of the Hamiltonians]\label{TheoIdentificationHamiltonians}
Under \autoref{AssumptionTrivial}, for $1\leq k\leq g$,
\[
 (2k-1)\operatorname{Ham}^{(\mathbf e_{2k-1})}
    (\boldsymbol\lambda,\boldsymbol\mu;D\mathbf s)
     =K_{2k-1}(\boldsymbol\lambda,\boldsymbol\mu;\mathbf s).
\]
Here $K_{2k-1}$ has exactly the normalization in
\autoref{PropTakasakiTheorem3}, with $\widetilde I_0$ in its first
interpolation formula.
\end{theorem}
\begin{proof}
By \autoref{LemmaEquality} and
$\boldsymbol\beta=C\widetilde{\mathbf R}$,
\begin{equation}\label{Proved1}
 [C^{-1}\mathbf T]_k
 =-\sum_{j=1}^g
      \frac{\mu_j\widetilde R_{k-1}'(\lambda_j)}
           {\beta_g'(\lambda_j)}.
\end{equation}
Consequently \autoref{PropTakasakiTheorem3} gives
\[
 \mathbf K=2\mathbf H+C^{-1}(\mathbf A+\mathbf T).
\]
This is exactly the vector in \autoref{TheoHamH}.
\end{proof}

One can alternatively define
\[
 K^\circ_{2k-1}:=K_{2k-1}-[C^{-1}\mathbf A]_k
\]
by replacing $\widetilde I_0$ with $I_0$ in the first interpolation
formula. With that different convention the comparison would read
\[
 (2k-1)\operatorname{Ham}^{(\mathbf e_{2k-1})}
       (\boldsymbol\lambda,\boldsymbol\mu;D\mathbf s)
    =K^\circ_{2k-1}+[C^{-1}\mathbf A]_k.
\]
These two normalizations give the same Hamiltonian vector fields,
but their formulas must not be mixed. We retain $K_{2k-1}$ throughout.

\begin{corollary}[Identification of the Hamiltonian evolutions]\label{CorollaryIdentification} Under \autoref{AssumptionTrivial}, the Hamiltonian evolutions of $(\mathbf{q},\mathbf{p})$ on the isomonodromic side are identical to the Hamiltonian evolutions of $(\boldsymbol{\lambda},\boldsymbol{\mu})$ on the $(2,2g+1)$ minimal model side upon the identification of times given by \autoref{TheoIdentificationTimes} and identification of the Darboux coordinates given by \autoref{TheoIdentificationDarbouxCoordinates}. 
\end{corollary}

\subsection{Symmetric Darboux coordinates}\label{SectionSym}
Symmetric Darboux coordinates $(\mathbf{Q},\mathbf{P}):=\left(Q_j,P_j\right)_{1\leq j\leq g}$ are introduced in \autoref{DefNewCoord} and correspond to the symmetric coordinates used by Mazzocco and Mo \cite{mazzocco2007hamiltonian} in the derivation of the Painlev\'{e} II hierarchy. An important feature of these coordinates is that they are related to the oper Darboux coordinates $\left(\mathbf{q},\mathbf{p}\right)$ by a time-independent and symplectic transformation as proved in \cite{MarchalP1Hierarchy}. In other words, we have:
\beq \sum_{j=1}^g dq_j\wedge dp_j= \sum_{j=1}^gdQ_j\wedge dP_j\eeq
\sloppy{A consequence of this property is that the Hamiltonians $\td{K}_{2k-1}(\mathbf{Q},\mathbf{P};\mathbf{s})$ (resp. $\td{\text{Ham}}^{(\mathbf{e}_{2k-1})}(\mathbf{Q},\mathbf{P};\mathbf{t})$) governing the Hamiltonian evolutions of $(\mathbf{Q},\mathbf{P})$ relative to $s_{2k-1}$ (resp. $t_{\infty,2k-1}$) can be obtained from $K_{2k-1}$ (resp. $\text{Ham}^{(\mathbf{e}_{2k-1})}$) by simply replacing the Darboux coordinates according to \autoref{DefNewCoord}. This observation was made by Takasaki in Section $11$ of \cite{Takasaki} but left as an open question there. The results of the present article combined with the proof of the symplectic property of the change of coordinates in \cite{MarchalP1Hierarchy} provide a proof of this observation.}

\medskip
In the direct isomonodromic approach developed in \cite{MarchalP1Hierarchy}, expressions for all relevant quantities are given in terms of these symmetric Darboux coordinates and some of them are reproduced (with adapted notation) in this article. For example, expressions of the Lax matrix $\hat{L}(\lambda)$ are given in \autoref{IsoPropSym}  and expressions of the Hamiltonians $\td{\text{Ham}}^{(\mathbf{e}_{2k-1})}(\mathbf{Q},\mathbf{P};\mathbf{t})$ are given in \autoref{HamQP}. The main advantage of these symmetric Darboux coordinates is that the corresponding Hamiltonians are polynomials in the Darboux coordinates (and quadratic in $\mathbf{P}$). Expressions of the auxiliary matrices $\hat{A}_{\boldsymbol{\alpha}}(\lambda)$ can also be found in \cite{MarchalP1Hierarchy}. To our knowledge, no general formula for the Hamiltonians $\td{K}_{2k-1}(\mathbf{Q},\mathbf{P};\mathbf{s})$ or for the Lax matrices in terms of $(\mathbf{Q},\mathbf{P})$ exists on the $(2,2g+1)$ minimal model side in the literature. Consequently, the identifications proposed in this article and results of \cite{MarchalP1Hierarchy} fill this gap. 

\medskip

Another point of interest is to clarify the role of $u$ in the isomonodromic formalism. Indeed, from \autoref{DefWaveMatrixFunction} and \autoref{PI-matrix-prop}, we have under \autoref{AssumptionTrivial} that
\beq \left[\mathcal{A}^{(g)}(\lambda)\right]_{2,1}=\lambda^{g+1} -\frac{u}{2}\lambda^{g}+ O\left(\lambda^{g-1}\right)\eeq
At the same time, \autoref{IsoPropSym} provides
\bea \left[\hat{L}(\lambda)\right]_{2,1}&=&-P^{(2)}_{\infty,2g+1}h_0(\{\mathbf{q}\})\lambda^{g+1}-\left( P^{(2)}_{\infty,2g+1}h_1(\{\mathbf{q}\})+P^{(2)}_{\infty,2g}h_0(\{\mathbf{q}\})\right)\lambda^g \cr&&
+O\left(\lambda^{g-1}\right)\cr
&=&\lambda^{g+1}+Q_1\lambda^{g}+O\left(\lambda^{g-1}\right)
\eea
since $P^{(2)}_{\infty,2g+1}=-1$ and $P^{(2)}_{\infty,2g}=0$ from \eqref{ReducedtdP2}. Thus, we have the identification:

\begin{theorem}[Identification of $u$ and $Q_1$]\label{IdentificationU}Under \autoref{AssumptionTrivial}, we have
\beq Q_1=q_1+\dots+q_g=-\frac{u}{2}\,\,\Leftrightarrow\,\, u=-2Q_1\eeq
\end{theorem}

This identifies the Lenard differential polynomials in $u$ with
differential polynomials in $Q_1$, with derivatives taken in
$x=s_1=t_{\infty,1}$. The normalized string equation is an ODE of order
$2g$ in $x$, with the other times as parameters. Eliminating those
parameters may produce higher-order necessary differential
consequences, as in \eqref{ODESecondMemberP1}; such consequences
need not be equivalent to the original hierarchy.

\medskip
The correction in the Hamiltonians has a short polynomial expression
in the symmetric Darboux coordinates.

\begin{proposition}[Correction in symmetric Darboux coordinates]\label{PropSymmCorrection}
Write
\[
 a(\lambda)=\hat L_{11}(\lambda)=\sum_{j=0}^{g-1}a_j\lambda^j,
 \qquad
 a_j=(-1)^j\sum_{i=j+1}^gP_iQ_{i-j-1},
 \qquad Q_0=1.
\]
Let $S_r(\{\mathbf q\})=\sum_iq_i^r$, with $S_0=g$.
For the vector $\mathbf T$ of \autoref{TheoHinftyI}, one has
\[
 T_1=0,\qquad
 T_k=(g-k+1)a_{g-k+1}
       -\sum_{j=g-k+1}^{g-1}
                   a_jS_{j-g+k-1}(\{\mathbf q\}),
 \quad 2\leq k\leq g.
\]
The phase-dependent correction in $K_{2k-1}$ is
\[
 [C^{-1}\mathbf T]_k
 =-\sum_{j=1}^g
       \frac{p_j\widetilde R_{k-1}'(q_j)}
            {\prod_{i\ne j}(q_j-q_i)}.
\]
It is polynomial in $\mathbf Q$ and linear in $\mathbf P$.
\end{proposition}
\begin{proof}
Put $b(\lambda)=\hat L_{12}(\lambda)=\prod_i(\lambda-q_i)$.
Lagrange interpolation gives
\[
 \frac{a(\lambda)}{b(\lambda)}
       =\sum_i\frac{p_i}{b'(q_i)(\lambda-q_i)},
 \qquad
 b\,\partial_\lambda(a/b)
       =-\sum_i\frac{p_i}{\lambda-q_i}
                    \prod_{j\ne i}\frac{\lambda-q_j}{q_i-q_j}.
\]
With the residue convention above, $T_k$ is the coefficient of
$\lambda^{g-k}$ in
\[
 b\,\partial_\lambda(a/b)=a'-a\,\frac{b'}b.
\]
Since
\[
 \frac{b'(\lambda)}{b(\lambda)}
       =\sum_{r\geq0}S_r(\{\mathbf q\})\lambda^{-r-1},
\]
coefficient extraction gives the displayed formula, including $T_1=0$.
Finally, use
$\boldsymbol\beta'=C\widetilde{\mathbf R}'$ and
\autoref{LemmaEquality}. The power sums are polynomials in
$\mathbf Q$ by \eqref{RelationSeReduced}.
\end{proof}

\section{Discussion} \label{sec6}
The main objective of this article was to provide a clear identification regarding the two approaches (isomonodromic or minimal model) of the PI hierarchy. This is achieved through \autoref{TheoIdentificationTimes}, \autoref{TheoIdentificationDarbouxCoordinates} and \autoref{TheoIdentificationHamiltonians}. As a by-product, it gives a non-trivial check of the validity of the formulas proposed by the two approaches. A natural consequence of this identification of the two formalisms is to be able to build a bridge between quantities that can be easy to understand on one side, but difficult on the other. For example, as discussed in \autoref{SectionSym}, the identification provides explicit formulas for the Lax matrices and Hamiltonians in terms of symmetric Darboux coordinates, which were missing in the minimal model side. 

The identification of both formalisms also provides interesting consequences.
\begin{itemize}
\item The isomonodromic formulation expresses the explicitly normalized
Hamiltonians as
\[
 K_{2k-1}(\boldsymbol\lambda,\boldsymbol\mu;\mathbf s)
 =(2k-1)\sum_{j=0}^{g-1}
     \nu_{\infty,j+1}^{(\mathbf e_{2k-1})}(D\mathbf s)\,
     H_{\infty,j}(\boldsymbol\lambda,\boldsymbol\mu;D\mathbf s),
 \qquad 1\leq k\leq g.
\]
The coefficients $H_{\infty,j}$ are independent of the chosen
deformation direction; the direction enters through $\nu$.
\item The auxiliary matrices can be written as
\[
 \hat A_{\mathbf e_{2k-1}}(\lambda)
  =\frac1{2k-1}\mathcal U_{2k-1}(\lambda)
  =\frac1{2k-1}
       \left(\left[\lambda^{k-1}U(\lambda)\right]_{\infty,+}
                       -R_{2k-1}[u]E_{21}\right),
 \qquad 1\leq k\leq g,
\]
with $u=-2Q_1$.
\end{itemize}
    
There are several open questions that arise from the present work. Indeed, in this paper, we could identify the direct isomonodromic approach describing the PI hierarchy with the corresponding reduction of the KP hierarchy. However, there exist generalizations on both sides, and it would be interesting to obtain similar results in these cases.
For example, the case of rank $2$ connections with arbitrary non-twisted singularities is done in \cite{marchal2024hamiltonianrepresentationisomonodromicdeformations} from the direct isomonodromic approach, while generalization of the present work to the reduction of the mKdV exists and provides the so-called Painlev\'{e} II hierarchy \cite{mazzocco2007hamiltonian}. In particular, this case should correspond to the case of one regular singularity at $\lambda=0$ and one irregular but non-twisted singularity at infinity and it would be interesting to test if formulas of \cite{marchal2024hamiltonianrepresentationisomonodromicdeformations} and \cite{mazzocco2007hamiltonian} are consistent and can be identified similarly to the work done in this paper. 

\medskip

Another direction is the generalization to any $(p,q)$ minimal models with positive coprime integers $p$ and $q$. In particular, the rank of the underlying connection is given by the integer $p$. On the minimal model side, there are existing results \cite{DiFrancesco:1990mc,Bonora:1994fq,Takasaki:1994xh,Marshakov:2009mn} and recent results covering specific rank $3$ cases \cite{Hayford2024}. Performing a similar analysis may allow for the understanding of how to generalize the rank two results of \cite{marchal2024hamiltonianrepresentationisomonodromicdeformations,MarchalP1Hierarchy} to higher rank cases which is currently an open question.

\section*{Acknowledgements} 
\sloppy{Olivier Marchal was supported by the fundamental junior IUF grant G752IUFMAR. Nathan Hayford was supported by the European Research Council (ERC), Grant Agreement No. 101002013. The authors would like to thank the organizers of the conference ``Integrability, Random Matrices, and All That" held at ENS Lyon in June 2025 where this work was initiated. }

\renewcommand{\theequation}{\thesection-\arabic{equation}}
\appendix

\section{String Equations and the KdV hierarchy}\label{Appendix-A}
In this appendix, we provide the necessary details in order to prove \autoref{PI-matrix-prop}. Virtually nothing we state here is new: all propositions are provided for clarity's sake, and most of them can be found in various places of the literature. Our main references here are the book \cite{Dickey}, and the article \cite{Takasaki}.

\subsection{Pseudodifferential operators}
We first review the basic facts about pseudodifferential operators that we will need.
Initially let $\mathbb A=\mathbb C[u,u_x,u_{xx},\ldots]$ be the
differential-polynomial algebra, with derivation $\partial_x$.
The square root of $Q$ and the Lenard polynomials have coefficients
in this algebra. In the dressing subsection we will explicitly
extend the coefficient algebra to include the required primitives.
\begin{definition}
    A \textit{pseudodifferential operator} over $\mathbb{A}$ is a formal expression of the form
        \begin{equation}
            X = \sum_{i\in\ZZ}X_i\partial^i,
        \end{equation}
    where $X_i\in \mathbb{A}$ and where it is assumed that there is an index $i_0\in \mathbb{Z}$ such that $X_{i}\equiv 0$ for all indices $i\geq i_0$. The order of a pseudodifferential operator is then defined as the index $i_1\in \mathbb{Z}$ such that $X_{i_1}\neq 0$ and $X_{i}=0$ for all $i>i_1$.
We set $\operatorname{ord}(0)=-\infty$.
\end{definition}
The symbol $\partial$ represents the derivative,
in the sense specified below. Coefficientwise addition and left
multiplication by elements of $\mathbb A$ make the space of
pseudodifferential operators an $\mathbb A$-module and, in particular,
a vector space over $\mathbb C$. We denote it by $\PsiDO$, or
$\mathbb A((\partial^{-1}))$. The algebra $\mathbb A$ is identified
with the multiplication operators. For $a\in\mathbb A$, the composition
of $\partial^k$ with multiplication by $a$ is defined by
    \begin{equation}
     \forall \, a\in \mathbb{A}, \,k\in \mathbb{Z}\,:\,   \partial^k \circ a := \sum_{j=0}^{\infty} \binom{k}{j} a^{(j)}(x)\partial^{k-j}.
    \end{equation}
Note that this definition makes sense also for \textit{negative integers} as pseudodifferential operators, and thus extends the usual Leibniz property. For example,
        \begin{equation*}
                \partial^{-1} \circ f = \sum_{k=0}^{\infty} (-1)^k f^{(k)} \partial^{-k-1} = f\partial^{-1} - f' \partial^{-2} + f''\partial^{-3} + ...
            \end{equation*}
Note also the relation $\partial^{-1} \circ \partial = \partial \circ \partial^{-1} = 1$. 
Finally, we can define multiplication, denoted  $\circ$, of elements of $\PsiDO$ by requiring that  $\partial^k \circ \partial^i=\partial^{k+i}$ for all $(k,i)\in \mathbb{Z}^2$ and that multiplication is distributive: $(X + Y)\circ Z = X\circ Z + Y\circ Z$ for all $(X,Y,Z)\in \PsiDO^3$. The multiplication on $\PsiDO$ is associative, in the sense that $(X\circ Y)\circ Z = X\circ (Y\circ Z)$, for any $X,Y,Z\in \PsiDO$.

We now define the purely differential (respectively, purely pseudodifferential) parts of a pseudodifferential operator.
\begin{definition}
    Any element $X\in \PsiDO$ can be split uniquely into two parts, a purely differential piece $X_+$ and a negative piece (purely pseudodifferential part) $X_-$:
    \begin{equation}
        X = \underbrace{\sum_{k\geq 0} X_k \partial^k}_{X_+} + \underbrace{\sum_{k< 0} X_k \partial^k}_{X_-}.
    \end{equation}
    We also define the \textit{residue} of a pseudodifferential operator as the coefficient of $\partial^{-1}$:
        \begin{equation}
            \Res_{\partial} X := X_{-1}.
        \end{equation}
\end{definition}
Finally, as it is relevant for the definition of the KdV hierarchy, we consider the operator
    \begin{equation}
        Q := \partial^2 + u.
    \end{equation}
    which is a purely differential operator. We have the following result.

\begin{proposition}
    There exists a unique pseudodifferential operator $Q^{1/2}\in \PsiDO$ such that $Q^{1/2}=\partial +\underset{k\leq 0}{\sum} Q_k \partial^k$ and
        \begin{equation}
            \left(Q^{1/2}\right)^2 = Q 
        \end{equation}
    The first few terms in the expansion of $Q^{1/2}$ are
        \begin{equation}
            Q^{1/2} = \partial + \frac{1}{2}u\partial^{-1} -\frac{1}{4}u_x\partial^{-2} + \frac{1}{8}\left(u_{xx}-u^2\right)\partial^{-3} + \mathcal{O}(\partial^{-4}).
        \end{equation}
\end{proposition} 
    \begin{proof}
        The proof here is standard, and can be found in \cite{Dickey}, for instance.
    \end{proof}

The pseudodifferential algebra is associative but noncommutative:
for example, $[\partial,u]=u_x$. An inverse of $X$ means a two-sided
inverse, and is denoted by $X^{-1}$ when it exists. Over the
differential-polynomial coefficient ring, not every nonzero operator
is invertible.

\begin{proposition}[Existence of inverse]\label{InverseProp}
If $X=1+N$ with $\operatorname{ord}N<0$, then $X$ has the unique
two-sided inverse
\[
 X^{-1}=\sum_{j=0}^\infty(-N)^j.
\]
\end{proposition}
\begin{proof}
Since $\operatorname{ord}N^j\leq-j$, every coefficient in this
series is a finite sum. The usual geometric-series telescoping
identity, involving powers of the single element $N$, gives both
$XX^{-1}=1$ and $X^{-1}X=1$. If $Y$ and $Z$ are two-sided inverses,
then $Y=Y(XZ)=(YX)Z=Z$.
\end{proof}

\subsection{The KdV hierarchy}
We are now ready to define the equations of the KdV hierarchy. Let $Q$ and $Q^{1/2}$ be the two pseudodifferential operators defined in the previous section, and set
    \begin{equation}\label{AppendixAdefB}
        \forall\, k\geq 1\,:\, B_k := \left(Q^{k/2}\right)_+ = \partial^k + O\left(\partial^{k-2}\right).
    \end{equation}
Note that $B_k$ is a purely differential operator of order $k$. We assume that $u$ is a (smooth) function of a collection of other indeterminates $\left(s_{2\ell+1}\right)_{\ell \geq 1}$, called \textit{times}. Moreover, we identify (for reasons that will soon become clear) $s_1:=x$. We also set ${\bf s} :=\left(s_1,s_3,s_5,\dots\right)$
\begin{definition}Let $\ell\geq0$. The $\ell^{\text{th}}$ equation of the KdV hierarchy is
    \begin{equation}\label{KdV-hierarchy}
        \frac{\partial Q}{\partial s_{2\ell+1}} = [B_{2\ell+1},Q].
    \end{equation}
\end{definition}

The first few equations in this hierarchy are
\begin{enumerate}
    \item $B_1 = \partial$, so the first equation of the hierarchy reads $\frac{\partial u}{\partial s_1} = u'$ and for this reason, it is convenient to identify $s_1$ with $x$.
    \item $B_3:=\left(Q^{3/2}\right)_+ = \partial^3  + \frac{3}{2}u\partial + \frac{3}{4}u'$. Thus, the second equation of the hierarchy then reads
        \begin{equation}\label{KdV}
            \frac{\partial u}{\partial s_3} = \frac{1}{4}u'''+ \frac{3}{2}uu';
        \end{equation}
    this is the \textit{KdV equation}.
\end{enumerate}
\begin{remark}
For $k\geq1$, one has $B_{2k}=Q^k$, so the associated even flow
would be trivial:
        \begin{equation*}
            \frac{\partial Q}{\partial s_{2k}} = [Q^{k},Q] = 0.
        \end{equation*}
    For this reason, we only consider the variables $\left(s_{2\ell+1}\right)_{\ell\geq 0}$ and not $\left(s_{2\ell}\right)_{\ell \geq 1}$.
\end{remark}
For $k\geq 0$, we call the vector fields $\frac{\partial}{\partial s_{2k+1}} - B_{2k+1}$ flows. Note that the operations of differentiation with respect to the parameter $s_{2k+1}$ and taking of positive parts commute:
    \begin{equation}
        \frac{\partial}{\partial s_{2k+1}} \left(X_+\right) = \left( \frac{\partial X}{\partial s_{2k+1}} \right)_+, \qquad\qquad X \in \PsiDO \,,\, k\geq 0.
    \end{equation}
As a consequence of the definition, all of the flows of the KdV hierarchy pairwise commute:
    \begin{proposition}
For any $\ell,m\geq0$,
            \begin{equation}
                \left[\frac{\partial}{\partial s_{2\ell+1}} - B_{2\ell+1}, \frac{\partial}{\partial s_{2m+1}} - B_{2m+1}\right] = 0.
            \end{equation}
    \end{proposition}
\begin{proof}
Put $L=Q^{1/2}$. The Lax equations imply
$\partial_{s_m}L^n=[B_m,L^n]$ for positive odd $m,n$.
Take positive parts and use $[L^m,L^n]=0$ together with closure
of the negative-order operators under commutators. This gives
\[
 \partial_{s_m}B_n-\partial_{s_n}B_m+[B_n,B_m]=0,
\]
which is precisely the stated commutation relation. The translation
flow $s_1=x$ is included.
\end{proof}
We can realize the equations of the KdV hierarchy alternatively as equations for the function $u(x)$. 

\begin{definition}\label{AppendixADefRk} Define the differential polynomials
    \bea R_{-1}[u]&:=&1\cr
       \forall \, \ell\geq 0\,:\, R_{2\ell+1}[u] &:=& \Res_{\partial}Q^{\ell+\frac{1}{2}}.
   \eea    
\end{definition}
The following proposition is from Takasaki \cite{Takasaki} (see equation 4.2):
\begin{proposition}
    The $\ell^{\text{th}}$ equation of the KdV hierarchy \eqref{KdV-hierarchy} is equivalent to the following equation on $u$:
        \begin{equation}
            \frac{\partial u}{\partial s_{2\ell+1}} = 2\frac{\partial}{\partial x} R_{2\ell+1}.
        \end{equation}
\end{proposition}
The coherence with \autoref{DefinitionRks} is given by the following result. 
\begin{proposition}[Lenard Recursion]
    The residues $R_{2\ell+1}$ satisfy the recursion relation (known as Lenard recursion):
        \begin{equation}
            R'_{2\ell+1} = \frac{1}{4} R_{2\ell-1}''' + uR_{2\ell-1}' + \frac{1}{2}u'R_{2\ell-1},
        \end{equation}
    subject to the initial condition $R_{-1}:=1$.
\end{proposition}
\begin{proof}
For $\ell=0$, the identity is
$R_1'=u_x/2$, which follows from $R_1=u/2$ and $R_{-1}=1$.
For $\ell\geq1$, put $k=2\ell-1$ and write
    \begin{equation}\label{EqSep}
        B_{k+2} = \left(Q^{k/2+1}\right)_+ = \left(\left((Q^{k/2})_+ + (Q^{k/2})_-\right)Q\right)_+ = \left(Q^{k/2}\right)_+Q + \left(\left(Q^{k/2}\right)_-Q\right)_+.
    \end{equation}
We can calculate this operator a little more explicitly:
    \bea\label{B-recursion}
        B_{k+2} &=& \left(Q^{k/2}\right)_+Q + \left(\left(R_k\partial^{-1} -\frac{1}{2}R_k'\partial^{-2} + O\left(\partial^{-3}\right)\right)(\partial^2 + u)\right)_+ \cr
        &=& \left(Q^{k/2}\right)_+ Q + R_k \partial -\frac{1}{2}R_k'.
    \eea
Taking the commutator of both sides of the above expression with $Q$, we obtain the identity
    \begin{equation}
        2R_{k+2}' = \left[\left(Q^{k/2+1}\right)_+,Q\right] = \left[\left(Q^{k/2}\right)_+ Q,Q\right] + \left[R_k \partial -\frac{1}{2}R_k',Q\right].
    \end{equation}
Expanding the right hand side, one obtains that
    \begin{equation}
        R_{k+2}' = \frac{1}{4} R_k''' + uR_k' + \frac{1}{2}u'R_k,
    \end{equation}
as desired.
\end{proof}

In particular, the previous proposition provides a recursive way to define the equations of the KdV hierarchy.

\subsection{The dressing method}
We now introduce the notion of \textit{dressing}, which will be useful in formulating the KdV hierarchy in matrix form.
For dressing, choose a differential-algebra extension
$\widehat{\mathbb A}\supset\mathbb A$ containing the successive
$x$-primitives required below. Equivalently, one may work locally
with coefficient functions and choose their primitives on a
simply connected $x$-domain. Write
$\widehat{\PsiDO}=\widehat{\mathbb A}((\partial^{-1}))$.

\begin{proposition}
There exists $\phi\in\widehat{\PsiDO}$ of the form
\[
 \phi=1+\sum_{k\geq1}w_k\partial^{-k}
\]
such that $Q^{1/2}=\phi\partial\phi^{-1}$.
It is unique up to right multiplication by
$1+\sum_{k\geq1}c_k\partial^{-k}$ with $\partial_xc_k=0$.
If the Sato equations in the other times are also imposed, the
remaining coefficients $c_k$ are constant in all those times.
\end{proposition}
\begin{proof}
Solve $Q^{1/2}\phi=\phi\partial$ by comparing successive
coefficients of $\partial^{-1},\partial^{-2},\ldots$.
At the first step,
\[
 \partial_xw_1=-\frac{u}{2},
\]
and each subsequent step determines an $x$-derivative of the next
$w_k$ in terms of the preceding coefficients. This explains why
primitives must be adjoined; the dressing coefficients need not lie
in $\mathbb A$ itself. The inverse exists by
\autoref{InverseProp}. If $\phi_1$ and $\phi_2$ are two dressings,
$\phi_1^{-1}\phi_2$ commutes with $\partial$, so its coefficients
are $x$-constant. The Sato equations similarly force the residual
right factor to be constant in the remaining times.
\end{proof}

We now define, for each $\ell\geq 1$,
    \begin{equation}\label{phi-evolution}
        \frac{\partial \phi}{\partial s_{2\ell+1}} := -\left(Q^{\ell+1/2}\right)_-\phi
    \end{equation}
We then have that
\begin{proposition}\label{phi-continuity}
The equations \eqref{phi-evolution} imply the equations of the KdV hierarchy \eqref{KdV-hierarchy}, if we identify $Q^{1/2} = \phi\partial\phi^{-1}$. Furthermore, the equations \eqref{phi-evolution} imply the commutativity conditions
        \begin{equation}
            \frac{\partial^2 \phi}{\partial s_{2\ell+1}\partial s_{2m+1}} = \frac{\partial^2 \phi}{\partial s_{2m+1}\partial s_{2\ell+1}},
        \end{equation}
    for every $(\ell,m)\in\left(\mathbb{N}\setminus\{0\}\right)^2$.
\end{proposition}

\begin{proof}
Differentiate $Q^{1/2}=\phi\partial\phi^{-1}$ and use
\eqref{phi-evolution}; the negative and positive parts give the
KdV Lax equations. The same calculation, together with their
zero-curvature identities, gives the commutativity of the
dressing evolutions. See also the dressing formulations in
\cite{ORLOV199377,Orlov:1986xz,Takasaki}.
\end{proof}

\begin{definition}
    \textit{Baker-Akhiezer/wave function.} Set the formal series in $z$:
        \begin{equation}
            \xi(z,{\bf s}) := \sum_{\ell=0}^{\infty}s_{2\ell+1}z^{2\ell+1},
        \end{equation}
    and define the action of $\partial^k$ on $e^{\xi(z,{\bf s})}$ by
        \begin{equation}\label{functional-map}
           \forall\, k\in\mathbb{Z}\,:\, \partial^ke^{\xi(z,{\bf s})} := z^k e^{\xi(z,{\bf s})}.
        \end{equation}

This action extends by linearity to pseudodifferential operators
over $\widehat{\mathbb A}$. The bare Baker--Akhiezer wave is
\begin{equation}
 \psi_{\rm BA}(z;\mathbf s)
 :=\phi e^{\xi(z,\mathbf s)}
 =\left(1+\frac{w_1}{z}+\frac{w_2}{z^2}
                    +O(z^{-3})\right)e^{\xi(z,\mathbf s)}.
\end{equation}
Until a spectral normalization is chosen below, we abbreviate
$\psi=\psi_{\rm BA}$.
\end{definition}
\begin{remark}
The symbol map is the $\widehat{\mathbb A}$-linear map
\[
 \widehat{\PsiDO}\longrightarrow
 \widehat{\mathbb A}((z^{-1}))\,e^{\xi(z,\mathbf s)},
 \qquad X\longmapsto X e^{\xi(z,\mathbf s)}.
\]
It is a linear map, not a homomorphism for operator composition.
\end{remark}

With this definition in place, we then have the following proposition:
    \begin{proposition}
        We have $Q = \phi \partial^2 \phi^{-1} = (Q^{1/2})^2$. Moreover
            \beq\label{Q-eq}
                Q\psi(z;{\bf s}) = z^2\psi(z;{\bf s}),\eeq
                \beq \label{s-eq}
                \frac{\partial}{\partial s_{2\ell+1}} \psi(z;{\bf s}) = B_{2\ell+1} \psi(z;{\bf s}),\qquad \forall \,\ell\in \mathbb{N}.
            \eeq
The $s_{2\ell+1}$ flow of \eqref{KdV-hierarchy} is equivalent
to compatibility of \eqref{Q-eq} and \eqref{s-eq}.
Compatibility between different equations \eqref{s-eq}
follows from the KdV zero-curvature identities and agrees with
\autoref{phi-continuity}.
    \end{proposition}
We again refer to \cite{Dickey} for the proof of this fact.

\begin{remark}
The equations
\[
 \partial_{s_{2\ell+1}}E=z^{2\ell+1}E,\qquad \ell\geq0,
\]
including the $x=s_1$ equation, determine
\[
 E(z,\mathbf s)=c(z)
       \exp\left(\sum_{\ell\geq0}s_{2\ell+1}z^{2\ell+1}\right)
\]
with a factor $c(z)$ independent of all the times.
They do not determine that factor as a numerical constant.
The bare wave uses the convention $c(z)=1$. The scalar spectral
normalization needed for the isomonodromic equation will be
chosen separately.
\end{remark}

\begin{remark}
    Since $Q = \partial^2 + u$, equation \eqref{Q-eq} defines an ODE in $x$ for the function $\psi(z;\bf s)$, for any fixed $z$. The compatibility of \eqref{Q-eq} with \eqref{s-eq} is the statement that the flows defined by the variables $s_{2\ell+1}$ are isospectral for $Q$.
\end{remark}

Due to the identity
    \begin{equation*}
B_{2\ell+1} =\left(Q^{\ell+1/2}\right)_+\overset{\eqref{EqSep}}{=} \left(Q^{\ell-1/2}\right)_+ Q + R_{2\ell-1} \partial -\frac{1}{2}R_{2\ell-1}',
    \end{equation*}
we can re-express the action of $B_{2\ell+1}$ in terms of the differential polynomials $R_{2\ell+1}$:
\begin{proposition}
    Fix $\ell\geq 1$, and define the polynomial in $z$
    \begin{equation}
        \mathcal{R}_{2\ell+1}(z) = z^{2\ell} + \sum_{j=1}^{\ell} R_{2j-1}z^{2(\ell-j)}.
    \end{equation}
Then, equation \eqref{s-eq} can be written alternatively as
    \begin{equation}
        \frac{\partial}{\partial s_{2\ell+1}}\psi(z;{\bf s}) = \mathcal{R}_{2\ell+1}(z)\psi_x(z;{\bf s}) -\frac{1}{2}\mathcal{R}_{2\ell+1,x}(z)\psi(z;{\bf s}).
    \end{equation}
\end{proposition}
\begin{proof}
    See Section 3 of \cite{Takasaki}.
\end{proof}
We now define the vector
    \begin{equation}
        \vec{\psi}(z;{\bf s}) :=
        \begin{pmatrix}
            \psi(z;{\bf s})\\
            \psi_x(z;{\bf s})
        \end{pmatrix}.
    \end{equation}
Using equation \eqref{Q-eq} to rewrite $\psi_{xx}(z;{\bf s}) = (z^2-u)\psi(z;{\bf s})$, we can rewrite equations \eqref{Q-eq} and \eqref{s-eq} as linear equations on the vector $\vec{\psi}(z;{\bf s})$:
    \begin{proposition}\label{KdV-flows-prop}
    Set $\lambda := z^2$.
        The equation \eqref{Q-eq} is equivalent to
            \begin{equation}\label{vec-Q}
                \frac{\partial}{\partial x}\vec{\psi}(z;{\bf s}) = \mathcal{U}_1(\lambda)\vec{\psi}(z;{\bf s}),
            \end{equation}
        and the equations \eqref{s-eq} are equivalent to
            \begin{equation}\label{vec-s}
                \frac{\partial}{\partial s_{2\ell+1}}\vec{\psi}(z;{\bf s}) = \mathcal{U}_{2\ell+1}(\lambda)\vec{\psi}(z;{\bf s}), \qquad \forall\, \ell \in \mathbb{N},
            \end{equation}
        The compatibility of equation \eqref{vec-Q} with \eqref{vec-s} is equivalent to the $\ell^{\text{th}}$ equation of the KdV hierarchy \eqref{KdV-hierarchy}. Here, the matrices 
        $\mathcal{U}_{2\ell+1}(\lambda)$ are defined by \autoref{DefWaveMatrixFunction}.
    \end{proposition}

\subsection{String Equations}
The string equations arise as additional symmetries of the KdV hierarchy. They can be formulated in terms of the Orlov-Schulman operator or as the $L_{-1}$ Virasoro constraint on the KdV $\tau$-function, but we prefer to give a more direct definition.

For each $g\geq 1$, we define the following differential operator:
    \begin{equation}\label{P-operator}
        P^{(g)} := B_{2g+1} + \sum_{\ell=1}^{g-1} \frac{2\ell+1}{2}s_{2\ell+1}B_{2\ell-1},
    \end{equation}
where the operators $B_{2\ell+1}$ are the generators of the KdV flows defined by \eqref{AppendixAdefB}.
    \begin{definition}
        The $(2,2g+1)$ string equation is defined to be the ODE
            \begin{equation}\label{operator-string-eq}
                [Q,P^{(g)}] = 1,
            \end{equation}
        which is supplemented by the first $g-1$ flows of the KdV hierarchy:
            \begin{equation}\label{operator-KdV-flows}
                \frac{\partial u}{\partial s_{2\ell+1}} = 2\frac{\partial}{\partial x} R_{2\ell+1}, \qquad \forall\, \ell \in \llbracket 1,g-1\rrbracket.
            \end{equation}
    \end{definition}
Since $[B_{2\ell+1},Q]=2\partial_xR_{2\ell+1}$,
\eqref{operator-string-eq} is equivalent to
$\partial_xF_g=0$, with $F_g$ as in \eqref{string-normalization}.

\begin{proposition}
For fixed values of the other times, the operator string equation
gives
\[
 R_{2g+1}[u]
  +\sum_{\ell=1}^{g-1}\frac{2\ell+1}{2}s_{2\ell+1}R_{2\ell-1}[u]
  +\frac{x}{2}=C,
 \qquad \partial_xC=0.
\]
The normalization used in this paper is $C=0$. At fixed remaining
times, it can be obtained by translating the origin of $x$.
\end{proposition}

The bare wave $\psi_{\rm BA}$ does not automatically satisfy an
exact spectral equation: its time-independent scalar spectral
factor has not yet been chosen.

\begin{proposition}[Spectral normalization of the wave]
Assume the normalized string equation and the KdV flows, and set
\[
 s_{2g+1}=0,\qquad
 s_{2g+3}=\frac{2}{2g+3},\qquad
 s_{2k+1}=0\quad(k\geq g+2).
\]
On a chosen formal branch $z^2=\lambda$, the bare wave can be
renormalized by a factor $\rho(z)$ independent of the active times:
\[
 \psi_{\rm iso}(z;\mathbf s)
     =\rho(z)\psi_{\rm BA}(z;\mathbf s),\qquad
 \rho(z)=z^{-1/2}(1+O(z^{-1})),
\]
up to a nonzero constant, so that
\begin{equation}\label{P-eq}
 P^{(g)}\psi_{\rm iso}
      =\partial_\lambda\psi_{\rm iso}
      =\frac1{2z}\partial_z\psi_{\rm iso}.
\end{equation}
The equations $Q\psi_{\rm iso}=\lambda\psi_{\rm iso}$ and
$\partial_{s_{2\ell+1}}\psi_{\rm iso}=B_{2\ell+1}\psi_{\rm iso}$
remain valid.
\end{proposition}
\begin{proof}
The operator identity $[Q,P^{(g)}]=1$ is the compatibility condition
for $Q\psi=\lambda\psi$ and $P^{(g)}\psi=\partial_\lambda\psi$.
Compatibility with the remaining time equations is the
operator identity proved in \autoref{extra-compat-prop}; that
proof is algebraic and does not assume a spectral normalization.

Reducing these scalar operators to the vector
$(\psi,\psi_x)^t$ gives a formally compatible rank-two system with
spectral matrix $\mathcal A^{(g)}$. Its ramified formal normal
form has a first component
\[
 \psi_{\rm iso}
   =z^{-1/2}(1+O(z^{-1}))
       \exp\left(\frac{2}{2g+3}z^{2g+3}
                    +\sum_{k=1}^g s_{2k-1}z^{2k-1}\right).
\]
The phase follows from \eqref{xiexp}, and the factor $z^{-1/2}$
is the logarithmic normalization of the traceless rank-two
system. The bare wave has the same exponential branch and
satisfies the same $x$ and active-time equations. Uniqueness of
the formal scalar solution in that branch therefore makes their
ratio independent of those times. This gives the stated $\rho$.
The highest-time value also follows directly by comparing
$\partial_\lambda\xi$ with the monic leading term
$z^{2g+1}$ of $P^{(g)}$.

For clarity, the normalization is necessary: the Wronskian of the
two bare waves at $z$ and $-z$ is
$-2z(1+O(z^{-1}))$, whereas a fundamental matrix for a traceless
spectral equation has $\lambda$-independent determinant.
\end{proof}

For the remainder of this appendix, write $\psi=\psi_{\rm iso}$
and $\vec\psi=(\psi,\psi_x)^t$. This changes no $x$ or active-time
equation. With this normalization, compatibility of
$Q\psi=\lambda\psi$ and $P^{(g)}\psi=\partial_\lambda\psi$ is
precisely \eqref{operator-string-eq}.
Furthermore, we have that
    \begin{proposition}\label{extra-compat-prop}
        Given equations \eqref{operator-string-eq}, \eqref{operator-KdV-flows}, we have that
            \begin{equation}\label{P-s-comp}
                \frac{\partial P^{(g)}}{\partial s_{2\ell+1}} = [B_{2\ell+1},P^{(g)}], \qquad \forall \, \ell \in \llbracket 1,g-1\rrbracket.
            \end{equation}
    \end{proposition}
\begin{proof}
Use the KdV evolutionary derivations before restricting the higher
times, and define
\[
 \mathfrak D
 =-\partial_{s_{2g+1}}
   -\sum_{\ell=1}^{g-1}
       \frac{2\ell+1}{2}s_{2\ell+1}\partial_{s_{2\ell-1}}.
\]
Since $\partial_{s_{2\ell-1}}Q=[B_{2\ell-1},Q]$,
the string equation gives $\mathfrak DQ=[Q,P^{(g)}]=1$.
These evolutionary derivations commute with $\partial_x$.
The leading-term recursion for $L=Q^{1/2}$ then gives
$\mathfrak DL=\tfrac12L^{-1}$: indeed, differentiating $L^2=Q$
gives $L(\mathfrak DL)+(\mathfrak DL)L=1$, and this equation
has a unique formal solution of negative order.
Consequently, for every integer $m$,
\[
 \mathfrak DQ^{m/2}=\frac m2Q^{(m-2)/2}.
\]
Taking positive parts gives
\begin{equation}\label{D-identity}
 \mathfrak DB_{2j+1}=\frac{2j+1}{2}B_{2j-1},\qquad j\geq1.
\end{equation}
For $1\leq j\leq g-1$, KdV zero curvature yields
\begin{align*}
 \partial_{s_{2j+1}}P^{(g)}
 &=\frac{2j+1}{2}B_{2j-1}
   +\partial_{s_{2j+1}}B_{2g+1}
   +\sum_{\ell=1}^{g-1}\frac{2\ell+1}{2}s_{2\ell+1}
                         \partial_{s_{2j+1}}B_{2\ell-1}\\
 &=\frac{2j+1}{2}B_{2j-1}
   +\partial_{s_{2g+1}}B_{2j+1}
   +[B_{2j+1},B_{2g+1}]\\
 &\quad+\sum_{\ell=1}^{g-1}\frac{2\ell+1}{2}s_{2\ell+1}
      \left(\partial_{s_{2\ell-1}}B_{2j+1}
                   +[B_{2j+1},B_{2\ell-1}]\right)\\
 &=\frac{2j+1}{2}B_{2j-1}
       -\mathfrak DB_{2j+1}+[B_{2j+1},P^{(g)}]\\
 &=[B_{2j+1},P^{(g)}].
\end{align*}
This proves \eqref{P-s-comp}.
\end{proof}

We can `upgrade' all of the above equations to matrix equations in the same way we did for the KdV flows. We summarize this in the following proposition.
\begin{proposition}\label{P-matrix-prop}
    The equation \eqref{P-eq} is equivalent to the matrix equation (again recalling that $\lambda = z^2$)
        \begin{equation}\label{P-matrix-eq}
            \frac{\partial}{\partial \lambda}\vec{\psi} = \mathcal{A}^{(g)}(\lambda)\vec{\psi} := \left(\mathcal{U}_{2g+1}(\lambda) + \sum_{\ell=1}^{g-1} \frac{2\ell+1}{2}s_{2\ell+1}\mathcal{U}_{2\ell-1}(\lambda)\right)\vec{\psi}.
        \end{equation}
    Compatibility of equation \eqref{P-matrix-eq} with \eqref{vec-Q} is equivalent to the string equation \eqref{operator-string-eq}. 
\end{proposition}
Finally, the following proposition follows immediately:
\begin{proposition}\label{Lemma3-prop}
Compatibility of \eqref{P-matrix-eq} with equations \eqref{vec-s} is equivalent to \eqref{P-s-comp}, and thus follows automatically from \eqref{operator-string-eq} and \eqref{operator-KdV-flows} due to \autoref{extra-compat-prop}.
\end{proposition}

\section{The Poisson bracket} \label{Appendix-B}
We recall the Poisson bracket used in \cite{Takasaki} to construct
the Hamiltonian structure of the PI hierarchy. Related brackets
and Mumford systems are discussed in \cite{EEKT,PedroniVanhaeke}.
The bracket is defined on the polynomial-matrix space of dimension
$3g+1$, with coordinates given by the coefficients of the entries
of $\mathcal A^{(g)}(\lambda)$.
Let us first recall that in the Saint Petersburg (tensorial) notation, the Poisson bracket in $\mathfrak{gl}_2(\mathbb{C})$ is expressed as:
\begin{equation}\label{bracket-def-1}
            \left\{\mathcal{A}^{(g)}(\lambda)\stackrel{\otimes}{,}\mathcal{A}^{(g)}(\mu)\right\} := \sum_{a,b,c,d\in \{1,2\}}\{\mathcal{A}^{(g)}_{ab}(\lambda),\mathcal{A}^{(g)}_{cd}(\mu)\} E_{ab}\otimes E_{cd}.
        \end{equation}
Thus, if we denote
\begin{equation}\label{NotAg}
        \mathcal{A}^{(g)}(\lambda) = 
        \begin{pmatrix}
            \alpha(\lambda) & \beta(\lambda)\\
            \gamma(\lambda) & -\alpha(\lambda)
        \end{pmatrix},
    \end{equation}
then we have the following proposition.

\begin{proposition}[Expression of the Poisson bracket with entries of the matrix]\label{prop-B-aa}Denoting $\sigma_{+}=E_{12}, \,\sigma_{-} = E_{21},\, \sigma_3 = \text{diag }(1,-1)$, then from \eqref{bracket-def-1} and \eqref{NotAg}, we have that
    \bea
        \left\{\mathcal{A}^{(g)}(\lambda)\stackrel{\otimes}{,}\mathcal{A}^{(g)}(\mu)\right\} &=& \{\alpha(\lambda),\alpha(\mu)\}\sigma_{3}\otimes\sigma_{3}+\{\alpha(\lambda),\beta(\mu)\}\sigma_3\otimes\sigma_+ +\{\alpha(\lambda),\gamma(\mu)\}\sigma_3\otimes\sigma_- \cr
        &&+ \{\beta(\lambda),\alpha(\mu)\}\sigma_+\otimes\sigma_3 + \{\beta(\lambda),\beta(\mu)\}\sigma_{+}\otimes\sigma_+ + \{\beta(\lambda),\gamma(\mu)\} \sigma_+\otimes\sigma_-\cr
        && + \{\gamma(\lambda),\alpha(\mu)\}\sigma_-\otimes\sigma_3 + \{\gamma(\lambda),\beta(\mu)\}\sigma_-\otimes\sigma_+ + \{\gamma(\lambda),\gamma(\mu)\}\sigma_-\otimes\sigma_- 
    \eea
\end{proposition}

The Poisson bracket is then specified using the following definition.

\begin{definition}[Definition of the Poisson bracket]\label{DefPoissonBracket}
The Poisson bracket is defined to be
    \bea\label{bracket-def-2}
        \left\{\mathcal{A}^{(g)}(\lambda)\stackrel{\otimes}{,}\mathcal{A}^{(g)}(\mu)\right\} &:=& \left[\mathcal{A}^{(g)}(\lambda)\otimes\mathbb{I} + \mathbb{I}\otimes\mathcal{A}^{(g)}(\mu),r(\lambda-\mu)\right]\cr &&+ \left[\mathcal{A}^{(g)}(\lambda)\otimes\mathbb{I} - \mathbb{I}\otimes\mathcal{A}^{(g)}(\mu),\Delta\right],
    \eea
    where $\Delta :=E_{21}\otimes E_{21}$, and $r(\lambda)$ is the standard ``$r$-matrix" on the loop algebra $\hat{\mathfrak{gl}}_2(\mathbb{C})$:
    \begin{equation}
        r(\lambda) = \frac{\mathcal{P}}{\lambda},
    \end{equation}
and $\mathcal{P} \in \text{End}(\CC^2\otimes\CC^2)$ is the linear operator defined on basis elements of $\CC^2\otimes \CC^2$ as
    \begin{equation}\label{perm-def}
        \mathcal{P}(v\otimes w) := w\otimes v,
    \end{equation}
The commutator is defined for any $(A,B,C,D)\in \left(\mathfrak{gl}_2(\mathbb{C})\right)^4$ by
\begin{equation*}
        [A\otimes B,C\otimes D] := (A\otimes B)(C\otimes D) - (C\otimes D)(A\otimes B) = AC\otimes BD - CA \otimes DB,
    \end{equation*}
following from the multiplication in the $\mathfrak{gl}_2(\mathbb{C})\otimes\mathfrak{gl}_2(\mathbb{C})$ algebra defined by
    \begin{equation*}
        (A\otimes B)(C\otimes D) := AC\otimes BD, \qquad \forall \, (A,B,C,D)\in \left(\mathfrak{gl}_2(\mathbb{C})\right)^4.
    \end{equation*}    
\end{definition}

\begin{remark}
Consider $G=I_2+fE_{21}$, where $f$ is independent of the spectral
parameter and Poisson-commutes with all phase-space functions.
For example, $f$ may depend only on the external times or Casimirs.
Then
\[
 \{G^{-1}\mathcal A^{(g)}(\lambda)G
       \stackrel{\otimes}{,}G^{-1}\mathcal A^{(g)}(\mu)G\}
 =(G\otimes G)^{-1}
   \{\mathcal A^{(g)}(\lambda)
       \stackrel{\otimes}{,}\mathcal A^{(g)}(\mu)\}
   (G\otimes G).
\]
Indeed, Poisson-centrality removes brackets involving $G$, while
$G\otimes G$ commutes with both the permutation operator and
$E_{21}\otimes E_{21}$. An arbitrary phase-dependent $f$ produces
additional Poisson terms and does not satisfy this assertion.
This covariance on the ambient matrix algebra is also distinct from
preservation of the normalized polynomial slice.
\end{remark}

One can compute the Poisson bracket to get the following proposition.

\begin{proposition}[Computation of the Poisson bracket]\label{prop-B-a} The Poisson bracket is given by
\begin{align*}
        \left\{\mathcal{A}^{(g)}(\lambda)\stackrel{\otimes}{,}\mathcal{A}^{(g)}(\mu)\right\} &=-2\left(\alpha(\lambda) - \alpha(\mu)\right) \sigma_{-}\otimes\sigma_- +   2\frac{\alpha(\lambda)-\alpha(\mu)}{\lambda-\mu}\left(\sigma_+\otimes\sigma_- - \sigma_-\otimes\sigma_+\right) \nonumber\\
        &+ \frac{\beta(\lambda)-\beta(\mu)}{\lambda-\mu}\left(\sigma_3\otimes\sigma_+ - \sigma_+\otimes\sigma_3\right) \\
        &- \frac{\gamma(\lambda)-\gamma(\mu)}{\lambda-\mu}\left(\sigma_3\otimes\sigma_- - \sigma_-\otimes\sigma_3\right) + \beta(\lambda)\sigma_3\otimes\sigma_- - \beta(\mu)\sigma_-\otimes\sigma_3. \nonumber
    \end{align*}
\end{proposition} 

\begin{proof}The proof follows from a series of computations that can be split into different lemmas. The proof is certainly not new but we could not find a reference where computations were made explicit and we decided to produce them for completeness.
\begin{lemma}
Let $\mathcal{P}$ be as defined in \eqref{perm-def}, then:
    \begin{align*}
        [\mathbb{I}\otimes \sigma_3,\mathcal{P}] &= -[\sigma_3\otimes\mathbb{I},\mathcal{P}] = 2\left(\sigma_-\otimes \sigma_+-\sigma_+\otimes \sigma_-\right),\\
        [\mathbb{I}\otimes \sigma_+,\mathcal{P}] &= -[\sigma_+\otimes \mathbb{I},\mathcal{P}] = \sigma_+\otimes\sigma_3 -\sigma_3\otimes \sigma_+,\\
        [\mathbb{I}\otimes \sigma_-,\mathcal{P}] &= -[\sigma_-\otimes\mathbb{I},\mathcal{P}] = \sigma_3\otimes \sigma_- - \sigma_-\otimes\sigma_3.
    \end{align*}
Let $\Delta = \sigma_-\otimes \sigma_-$. Then:
    \begin{align*}
        [\mathbb{I}\otimes \sigma_3,\Delta] &= [\sigma_3\otimes \mathbb{I},\Delta] = -2\sigma_-\otimes \sigma_-,\\
        [\mathbb{I}\otimes \sigma_+,\Delta] &= \sigma_-\otimes \sigma_3, \qquad [\sigma_+\otimes \mathbb{I},\Delta] = \sigma_{3}\otimes \sigma_-,\\
        [\mathbb{I}\otimes \sigma_-,\Delta] &= [\sigma_-\otimes \mathbb{I},\Delta] = 0.
    \end{align*}
\end{lemma}
The proof of the above follows from elementary (albeit tedious) calculations. The fundamental identity is
    \begin{equation}
        [E_{ij}\otimes E_{k\ell},E_{ab}\otimes E_{cd}] = \delta_{aj}\delta_{\ell c} E_{ib}\otimes E_{kd} - \delta_{ib}\delta_{kd} E_{aj}\otimes E_{c\ell}.
    \end{equation}
The next lemma that we shall need is the following:
\begin{lemma} \label{PB-B1} We have
    \begin{align}
        \left[\mathcal{A}^{(g)}(\lambda) \otimes \mathbb{I} + \mathbb{I}\otimes \mathcal{A}^{(g)}(\mu),r(\lambda-\mu) \right] &= 2\frac{\alpha(\lambda)-\alpha(\mu)}{\lambda-\mu}\left(\sigma_+\otimes\sigma_- - \sigma_-\otimes\sigma_+\right) \nonumber\\
        &+ \frac{\beta(\lambda)-\beta(\mu)}{\lambda-\mu}\left(\sigma_3\otimes\sigma_+ - \sigma_+\otimes\sigma_3\right) \\
        &- \frac{\gamma(\lambda)-\gamma(\mu)}{\lambda-\mu}\left(\sigma_3\otimes\sigma_- - \sigma_-\otimes\sigma_3\right). \nonumber
    \end{align}
\end{lemma}
\begin{proof}
Write $\mathcal{A}^{(g)}(\lambda) = \alpha(\lambda)\sigma_{3} + \beta(\lambda)\sigma_+ + \gamma(\lambda)\sigma_-$. Then, by the previous lemma,
    \bea
         [\mathcal{A}^{(g)}(\lambda) \otimes \mathbb{I},r(\lambda-\mu) ] &=& \frac{\alpha(\lambda)}{\lambda-\mu}[\sigma_{3}\otimes\mathbb{I},\mathcal{P}] + \frac{\beta(\lambda)}{\lambda-\mu}[\sigma_{+}\otimes\mathbb{I},\mathcal{P}] + \frac{\gamma(\lambda)}{\lambda-\mu}[\sigma_{-}\otimes\mathbb{I},\mathcal{P}]\cr
         &=& \frac{2\alpha(\lambda)}{\lambda-\mu}\left(\sigma_+\otimes \sigma_- -\sigma_-\otimes \sigma_+\right)
         +\frac{\beta(\lambda)}{\lambda-\mu}\left(\sigma_3\otimes \sigma_+-\sigma_+\otimes\sigma_3\right)\cr&&
         -\frac{\gamma(\lambda)}{\lambda-\mu}\left(\sigma_3\otimes \sigma_- -\sigma_-\otimes\sigma_3\right)
    \eea
    \bea
    [\mathbb{I}\otimes\mathcal{A}^{(g)}(\mu),r(\lambda-\mu) ] &=& \frac{\alpha(\mu)}{\lambda-\mu}[\mathbb{I}\otimes\sigma_{3},\mathcal{P}] + \frac{\beta(\mu)}{\lambda-\mu}[\mathbb{I}\otimes\sigma_{+},\mathcal{P}] + \frac{\gamma(\mu)}{\lambda-\mu}[\mathbb{I}\otimes\sigma_{-},\mathcal{P}]\cr
    &=&-\frac{2\alpha(\mu)}{\lambda-\mu}\left(\sigma_+\otimes \sigma_- - \sigma_-\otimes \sigma_+\right)
    -\frac{\beta(\mu)}{\lambda-\mu} \left(\sigma_3\otimes \sigma_+ - \sigma_+\otimes\sigma_3\right)\cr&&
    +\frac{\gamma(\mu)}{\lambda-\mu}\left(\sigma_3\otimes \sigma_- - \sigma_-\otimes\sigma_3\right)
    \eea
Combining these formulae yields the result of \autoref{PB-B1}.
\end{proof}
Finally the last lemma that we need is the following:
\begin{lemma}\label{PB-B2}We have
    \begin{equation}
        \left[\mathcal{A}^{(g)}(\lambda) \otimes \mathbb{I} - \mathbb{I}\otimes \mathcal{A}^{(g)}(\mu), \Delta\right] = -2\left(\alpha(\lambda) - \alpha(\mu)\right) \sigma_{-}\otimes\sigma_- + \beta(\lambda)\sigma_3\otimes\sigma_- - \beta(\mu)\sigma_-\otimes\sigma_3.
    \end{equation}
\end{lemma}
\begin{proof}
    Again writing $\mathcal{A}^{(g)}(\lambda) = \alpha(\lambda)\sigma_{3} + \beta(\lambda)\sigma_+ + \gamma(\lambda)\sigma_-$, we have that
        \begin{align*}
             [\mathcal{A}^{(g)}(\lambda) \otimes \mathbb{I},\Delta] &= \alpha(\lambda)[\sigma_3\otimes\mathbb{I},\Delta] + \beta(\lambda)[\sigma_+\otimes\mathbb{I},\Delta] + \gamma(\lambda)[\sigma_-\otimes\mathbb{I},\Delta]\\
             &=-2\alpha(\lambda)\sigma_-\otimes\sigma_- + \beta(\lambda)\sigma_3\otimes\sigma_-,
        \end{align*}
        and
        \begin{align*}
             [\mathbb{I}\otimes\mathcal{A}^{(g)}(\mu),\Delta] &= \alpha(\mu)[\mathbb{I}\otimes\sigma_3,\Delta] + \beta(\mu)[\mathbb{I}\otimes\sigma_+,\Delta] + \gamma(\mu)[\mathbb{I}\otimes\sigma_-,\Delta]\\
             &=-2\alpha(\mu)\sigma_-\otimes\sigma_- + \beta(\mu)\sigma_-\otimes\sigma_3.
        \end{align*}
    Combining these two results yields the result of \autoref{PB-B2}.
\end{proof}
Eventually, combining \autoref{PB-B1} and \autoref{PB-B2} provides the proof of \autoref{prop-B-a}.
\end{proof}

The explicit expression of the Poisson bracket given by \autoref{prop-B-a} combined with the tensorial expression of \autoref{prop-B-aa} gives the following result.

\begin{proposition}\label{prop-B1}We have the following scalar Poisson brackets:
        \begin{align*}
         \{\alpha(\lambda),\alpha(\mu)\} &=  \{\beta(\lambda),\beta(\mu)\} = 0,\\
         \{\alpha(\lambda),\beta(\mu)\} &= \frac{\beta(\lambda)-\beta(\mu)}{\lambda-\mu},\qquad \{\alpha(\lambda),\gamma(\mu)\} = -\frac{\gamma(\lambda)-\gamma(\mu)}{\lambda-\mu} +\beta(\lambda),\\
         \{\beta(\lambda),\gamma(\mu)\} &= 2\frac{\alpha(\lambda)-\alpha(\mu)}{\lambda-\mu},\qquad\qquad \{\gamma(\lambda),\gamma(\mu)\} = -2\left(\alpha(\lambda)-\alpha(\mu)\right).
    \end{align*}
\end{proposition}

These scalar relations imply, after straightforward computations that the following results hold.

\begin{proposition}
        \begin{equation}
        \{\mathcal{A}^{(g)}(\lambda),h^{(g)}(\mu)\} := 
        \begin{pmatrix}
            \{\alpha(\lambda),h^{(g)}(\mu)\} & \{\beta(\lambda),h^{(g)}(\mu)\}\\
             \{\gamma(\lambda),h^{(g)}(\mu)\} & -\{\alpha(\lambda),h^{(g)}(\mu)\}
        \end{pmatrix} = \left[\mathcal{A}^{(g)}(\lambda),\frac{\mathcal{A}^{(g)}(\mu)}{\lambda-\mu} + \beta(\mu)\sigma_-\right].
    \end{equation}
\end{proposition}
Extracting coefficients in the preceding identity gives the
following result.

\begin{proposition}\label{Poisson-Ik-prop}
For $0\leq\ell\leq g-1$,
\[
 \{\mathcal A^{(g)}(\lambda),I_{\ell+1}\}
       =[A_\ell^{(g)}(\lambda),\mathcal A^{(g)}(\lambda)],
\]
where
\[
 A_\ell^{(g)}(\lambda)
   =[\lambda^{\ell-g}\mathcal A^{(g)}(\lambda)]_{\infty,+}
                                  -\beta_{g,\ell+1}E_{21}.
\]
Moreover,
\begin{equation}\label{casimir}
 \{\mathcal A^{(g)}(\lambda),\widetilde I_0(\mu)\}=0.
\end{equation}
\end{proposition}
\begin{proof}
Expand the identity for
$\{\mathcal A^{(g)}(\lambda),h^{(g)}(\mu)\}$ at $\mu=\infty$.
The coefficient of $\mu^{g-\ell-1}$ in
$\mathcal A^{(g)}(\mu)/(\lambda-\mu)$ is
$-[\lambda^{\ell-g}\mathcal A^{(g)}(\lambda)]_{\infty,+}$,
and its coefficient in $\beta_g(\mu)E_{21}$ is
$\beta_{g,\ell+1}E_{21}$. This gives the first identity; replacing
the coefficient $I_{\ell+1}+A_{\ell+1}$ by $I_{\ell+1}$ makes no
difference on a fixed time leaf. The coefficients of degrees
at least $g$ cancel or vanish, giving \eqref{casimir}.
\end{proof}

\begin{proposition}
For $0\leq\ell\leq g-1$, with the normalization of
\autoref{DefSpectralHamiltonian},
\[
 \{\mathcal A^{(g)}(\lambda),2H_{\ell+1}\}
       =[\mathcal U_{2\ell+1}(\lambda),\mathcal A^{(g)}(\lambda)].
\]
Consequently the nonautonomous Lax equation is
\[
 \frac{\partial\mathcal A^{(g)}}{\partial s_{2\ell+1}}
   =2\{\mathcal A^{(g)},H_{\ell+1}\}
                        +\partial_\lambda\mathcal U_{2\ell+1}.
\]
\end{proposition}
\begin{proof}
The definitions of $A_\ell^{(g)}$, $\mathcal U_{2j+1}$ and
$\beta_{g,\ell+1}$ give
\[
 A_\ell^{(g)}
       =\sum_{j=0}^\ell c_{2g+1-2(\ell-j)}\mathcal U_{2j+1}.
\]
Combine this Toeplitz relation with $\mathbf I=2C\mathbf H$,
taking the times as fixed under Poisson brackets. The last formula
then follows from the spectral/time zero-curvature equation.
\end{proof}
Finally, let us show that this Poisson bracket gives rise to a set of canonical coordinates. As detailed in \cite{Takasaki}, one can produce a set of canonical coordinates by defining
    \begin{equation}
\beta(\lambda_i) = 0, \qquad \mu_i = \alpha(\lambda_i),\qquad  \forall \, i \in \llbracket 1,g\rrbracket.
    \end{equation}
These are just the coordinates of \autoref{DefinitionSpectralCoord} written in the present notations.
We then have that
\begin{proposition}\label{Canonical-proof}
On the open subset where $\beta$ has $g$ distinct roots, choose a
local ordering and set $\mu_i=\alpha(\lambda_i)$. The bracket
\eqref{bracket-def-2} gives
\[
 \{\lambda_i,\lambda_j\}=\{\mu_i,\mu_j\}=0,\qquad
 \{\lambda_i,\mu_j\}=\delta_{ij},\qquad 1\leq i,j\leq g.
\]
\end{proposition}
\begin{proof}
The identity $\{\beta(\lambda),\beta(\eta)\}=0$, differentiated
through the simple-root relations, gives
$\{\lambda_i,\lambda_j\}=0$. The relation
\[
 \{\alpha(\lambda),\beta(\eta)\}
                  =\frac{\beta(\lambda)-\beta(\eta)}{\lambda-\eta}
\]
then yields
\[
 \{\mu_i,\beta(\eta)\}
              =\frac{\beta(\eta)}{\eta-\lambda_i}.
\]
Evaluating the total bracket with $\beta(\lambda_j)=0$ gives
\[
 0=\delta_{ij}\beta'(\lambda_j)
                       +\beta'(\lambda_j)\{\mu_i,\lambda_j\}.
\]
Since $\beta'(\lambda_j)\ne0$,
$\{\lambda_j,\mu_i\}=\delta_{ij}$.

For the last relation, the same scalar bracket implies
\[
 \{\lambda_i,\alpha(\eta)\}
       =\frac{\beta(\eta)}
              {\beta'(\lambda_i)(\eta-\lambda_i)}.
\]
Using this formula and $\{\alpha(\lambda),\alpha(\eta)\}=0$,
the chain rule gives
\[
 \{\mu_i,\mu_j\}
     =\alpha'(\lambda_i)\delta_{ij}
                       -\alpha'(\lambda_j)\delta_{ij}=0.
\]
\end{proof}

\bibliographystyle{plain}
\bibliography{Biblio_AIHP}

@article{AtiyahBott,
  author  = {M. F. Atiyah and R. Bott},
  title   = {The {Yang--Mills} equations over {Riemann} surfaces},
  journal = {Philos. Trans. Roy. Soc. London Ser. A},
  volume  = {308},
  number  = {1505},
  pages   = {523--615},
  year    = {1983},
  doi     = {10.1098/rsta.1983.0017}
}

@misc{BoalchYamakawa,
	Author = {P. Boalch and D. Yamakawa},
	Note = {arXiv:1512.08091},
	Title = {Twisted wild character varieties},
	Year={2015},
	}

@article{JimboMiwaUeno,
	Author = {M.~Jimbo and T.~Miwa and K. Ueno},
	Fjournal = {Physica D. Nonlinear Phenomena},
	Journal = {Phys. D},
	Number = {2},
	Pages = {306--352},
	Title = {Monodromy preserving deformation of linear ordinary differential equations with rational coefficients: I. General theory and $\tau$-function},
	Volume = {2},
	Year = {1981}}

@article{JimboMiwa,
	Author = {M.~Jimbo and T.~Miwa},
	Doi = {10.1016/0167-2789(81)90021-X},
	Fjournal = {Physica D. Nonlinear Phenomena},
	Issn = {0167-2789},
	Journal = {Phys. D},
	Mrclass = {34A20 (14K25 58A15 58F07 81C05)},
	Mrnumber = {625446},
	Mrreviewer = {V. A. Golubeva},
	Number = {3},
	Pages = {407--448},
	Title = {Monodromy preserving deformation of linear ordinary differential equations with rational coefficients. {II}},
	Url = {https://doi.org/10.1016/0167-2789(81)90021-X},
	Volume = {2},
	Year = {1981}}

@article{Boalch2012,
	Author = {P. Boalch},
	Fjournal = {Publications Math\'{e}matiques de l'IH\'{E}S},
	Journal = {Publ. Math. IH\'{E}S},
	Pages = {1--68},
	Title = {Simply-laced isomonodromy systems},
	Volume = {116},
	Year = {2012}}

@article{Boalch2001,
	Author = {P. Boalch},
	Fjournal = {Advances in Mathematics},
	Journal = {Adv. Math.},
	Number = {2},
	Pages = {137--205},
	Title = {Symplectic Manifolds and Isomonodromic Deformations},
	Volume = {163},
	Year = {2001}}

@article{mazzocco2007hamiltonian,
	Author = {M. Mazzocco and M. Mo},
	Journal = {Nonlinearity},
	Number = {12},
	Pages = {2845--2882},
	Publisher = {IOP Publishing},
    doi = {10.1088/0951-7715/20/12/006},
	Title = {The {H}amiltonian structure of the second {P}ainlev{\'e} hierarchy},
	Volume = {20},
	Year = {2007}}

@article{MarchalAlameddineIsospectralIsomono2023,
  author  = {O. Marchal and M. Alameddine},
  title   = {Isomonodromic and isospectral deformations of meromorphic
             connections: the $\mathfrak{sl}_2(\mathbb C)$ case},
  journal = {Nonlinearity},
  volume  = {37},
  number  = {11},
  pages   = {115006},
  year    = {2024},
  doi     = {10.1088/1361-6544/ad7b96}
}

@book{FromGaussToPainleve,
  title={From {G}auss to {P}ainlev{\'e}},
  author={K. Iwasaki and H. Kimura and S. Shimomura and M. Yoshida},
  year={1991},
	publisher={Springer Vieweg-Teubner},
	series={Aspects of Mathematics},
	volume={16},
  isbn={978-3-322-90165-1}
}

@article{Malmquist1922,
	Author = {J. Malmquist},
	Fjournal = {Arkiv f\"{o}r matematik, astronomi och fysik},
	Journal = {Ark. {M}at. {A}str. {F}ys.},
	Pages = {1--89},
	Title = {Sur les \'{e}quations diff\'{e}rentielles du second ordre dont l'int\'{e}grale g\'{e}n\'{e}rale a ses points critiques fixes},
	Volume = {17},
	Year = {1922}
}

@article{Okamoto1986Iso,
title={Isomonodromic deformation and {P}ainlev\'{e} equations, and the {G}arnier system},
 volume={33}, 
 journal={J. Fac. Sci. Univ. Tokyo},
 author={K. Okamoto},
 year={1986},
 pages={575--618}
}

@article{marchal2024hamiltonianrepresentationisomonodromicdeformations,
      title={Hamiltonian representation of isomonodromic deformations of general rational connections on $\mathfrak{gl}_2(\mathbb{C})$}, 
      author={O. Marchal and N. Orantin and M. Alameddine},
    doi = "10.1088/1361-6544/ae19dd",
    journal = "Nonlinearity",
    volume = "38",
    number = "11",
    pages = "115018",
    year = "2025"
}

@article{Takasaki,
  author  = {K. Takasaki},
  title   = {Hamiltonian structure of {PI} hierarchy},
  journal = {SIGMA},
  volume  = {3},
  pages   = {042},
  note    = {32 pp.},
  year    = {2007},
  doi     = {10.3842/SIGMA.2007.042}
}

@article{MarchalP1Hierarchy,
  author  = {O. Marchal and M. Alameddine},
  title   = {Hamiltonian representation of isomonodromic deformations
             of twisted rational connections:
             the {Painlev\'e} 1 hierarchy},
  journal = {Comm. Math. Phys.},
  volume  = {406},
  pages   = {12},
  year    = {2025},
  doi     = {10.1007/s00220-024-05187-0}
}

@book{Dickey,
title={Soliton Equations and Hamiltonian Systems},
author={L. Dickey},
publisher={World Scientific},
series={Advanced Series in Mathematical Physics},
volume={26},
edition={Second},
year={2003},
address = {},
doi ={10.1142/5108}
}

@misc{MohamadP1Isospectral,
      title={A {T}wisted $\text{sl}_2(\mathbb{C})$ {I}somonodromic-{I}sospectral {C}orrespondence}, 
      author={Alameddine, M.},
      year={2025},
      eprint={2507.06668},
      archivePrefix={arXiv},
      primaryClass={math-ph},
      note={arXiv:2507.06668}
}

@article{Brezin:1990rb,
    author = "Brezin, E. and Kazakov, V. ",
    title = "{Exactly {S}olvable {F}ield {T}heories of {C}losed {S}trings}",
    doi = "10.1016/0370-2693(90)90818-Q",
    journal = "Phys. Lett. B",
    volume = "236",
    pages = "144--150",
    year = "1990"
}

@article{DOUGLAS1990176,
title = {Strings in less than one dimension and the generalized {K}d{V} hierarchies},
journal = {Phys. Lett. B},
volume = {238},
number = {2},
pages = {176--180},
year = {1990},
issn = {0370-2693},
doi = {https://doi.org/10.1016/0370-2693(90)91716-O},
url = {https://www.sciencedirect.com/science/article/pii/037026939091716O},
author = {M.R. Douglas}
}

@article{Dubrovin1976,
  title = {NON-LINEAR EQUATIONS OF {K}ORTEWEG-DE {V}RIES TYPE,  FINITE-ZONE LINEAR OPERATORS,  AND ABELIAN VARIETIES},
  volume = {31},
  ISSN = {1468-4829},
  url = {http://dx.doi.org/10.1070/RM1976v031n01ABEH001446},
  DOI = {10.1070/rm1976v031n01abeh001446},
  number = {1},
  journal = {Russian Math. Surv.},
  publisher = {Steklov Mathematical Institute},
  author = {Dubrovin, B. and Matveev, V. and Novikov, S. },
  year = {1976},
  pages = {59--146}
}

@article{Krichever1977,
  title = {METHODS OF ALGEBRAIC GEOMETRY IN THE THEORY OF NON-LINEAR EQUATIONS},
  volume = {32},
  ISSN = {1468-4829},
  url = {http://dx.doi.org/10.1070/RM1977v032n06ABEH003862},
  DOI = {10.1070/rm1977v032n06abeh003862},
  number = {6},
  journal = {Russian Math. Surv.},
  publisher = {Steklov Mathematical Institute},
  author = {Krichever, I. },
  year = {1977},
  pages = {185--213}
}

@article{Segal1985,
  title = {Loop groups and equations of {K}d{V} type},
  volume = {61},
  ISSN = {1618-1913},
  url = {http://dx.doi.org/10.1007/BF02698802},
  DOI = {10.1007/bf02698802},
  number = {1},
  journal = {Publ. Math. IHÉS},
  publisher = {Springer Science and Business Media LLC},
  author = {Segal,  G. and Wilson,  G.},
  year = {1985},
  pages = {5--65}
}

@incollection{sato1983soliton,
  title={Soliton equations as dynamical systems on infinite dimensional Grassmann manifold},
  author={Sato, M.},
  booktitle={North-Holland Mathematics Studies},
  volume={81},
  pages={259--271},
  year={1983},
  publisher={Elsevier}
}

@article{Date:1982tj,
    author = "Date, E. and Jimbo, M. and Kashiwara, M. and Miwa, T.",
    title = "{Transformation groups for soliton equations IV. A new hierarchy of soliton equations of KP-type}",
    fjournal={Physica D: Nonlinear Phenomena},
    journal={Phys. D},
    volume = {4},
    number={3},
    pages={343--365},
    year = "1982"
}

@incollection{Mumford2007,
  title = {The translation-invariant vector fields},
  ISBN = {9780817645786},
  DOI = {10.1007/978-0-8176-4578-6_4},
  booktitle = {{Tata Lectures on Theta II}},
  publisher = {{Birkh\"{a}user Boston}},
  author = {Mumford,  D.},
  year = {2007},
  pages = {40--50}
}

@misc{Hayford2024,
      title={{The Ising Model Coupled to $2D$ Gravity: Higher-order Painlev\'{e} Equations/The $(3,4)$ String Equation}}, 
      author={N. Hayford},
      year={2024},
      eprint={2405.03260},
      archivePrefix={arXiv},
      primaryClass={math-ph},
      note={arXiv:2405.03260}
}

@inproceedings{DiFrancesco:1990mc,
    author = "Di Francesco, P. and Kutasov, D.",
    title = "{Integrable models of two-dimensional quantum gravity}",
    booktitle = "{Cargese Study Institute: Random Surfaces, Quantum Gravity and Strings}",
    reportNumber = "PUPT-1206",
    month = "10",
    year = "1990"
}

@article{Marshakov:2009mn,
    author = "Marshakov, A.",
    title = "{On two-dimensional quantum gravity and quasiclassical integrable hierarchies}",
    doi = "10.1088/1751-8113/42/30/304021",
    journal = "J. Phys. A",
    volume = "42",
    number = "30",
    pages = "304021",
    year = "2009"
}

@article{Bonora:1994fq,
    author = "Bonora, L. and Liu, Q. P. and Xiong, C. S.",
    title = "{The Integrable hierarchy constructed from a pair of higher KdV hierarchies and its associated W algebra}",
    doi = "10.1007/BF02101629",
    journal = "Commun. Math. Phys.",
    volume = "175",
    pages = "177--202",
    year = "1996"
}

@article{Takasaki:1994xh,
    author = "K. Takasaki and T. Takebe",
    title = "{Integrable hierarchies and dispersionless limit}",
    doi = "10.1142/S0129055X9500030X",
    journal = "Rev. Math. Phys.",
    volume = "7",
    pages = "743--808",
    year = "1995"
}

@article{Orlov:1986xz,
    author = {A.Y. Orlov and E.I. Schulman},
    title = "{Additional Symmetries for Integrable Equations and Conformal Algebra Representation}",
    doi = "10.1007/BF00416506",
    journal = "Lett. Math. Phys.",
    volume = "12",
    pages = "171--179",
    year = "1986"
}

@article{ORLOV199377,
title = {{Dressing method, Darboux transformation and generalized restricted flows for the KdV hierarchy}},
fjournal = {Physica D: Nonlinear Phenomena},
journal = {Phys. D},
volume = {69},
number = {1},
pages = {77--84},
year = {1993},
issn = {0167-2789},
doi = {https://doi.org/10.1016/0167-2789(93)90181-Y},
author = {A.Y. Orlov and S. Rauch-Wojciechowski},
}

@article{Dubrovin-06,
title = {{On Hamiltonian Perturbations of Hyperbolic Systems of Conservation Laws, II: Universality of Critical Behaviour}},
author = {B. Dubrovin},
journal = {Comm. Math. Phys.},
volume ={267},
year = {2006},
pages = {117--139}
}

@article{Dubrovin-Grava-Klein,
title = {{On universality of critical behaviour in the focusing nonlinear Schr\"{o}dinger equation, elliptic umbilic catastrophe and the tritronqu\'{e}e solution to the Painlev\'{e}-I equation}},
author = {B. Dubrovin and T. Grava and C. Klein},
journal = {J. Nonlinear Sci.},
volume = {19},
number = {1},
year = {2009},
pages = {57--94}
}

@article{GM1,
title = {{A Non-Perturbative Treatment of Two-Dimensional Quantum Gravity}},
author = {D.J. Gross and A.A. Migdal},
journal = {Nucl. Phys. B},
volume = {340},
number = {2--3},
year = {1990},
pages = {333--365}
}

@article{GM2,
title = {Nonperturbative two-dimensional quantum gravity},
author = {D.J. Gross and A.A. Migdal},
journal = {Phys. Rev. Lett.},
volume = {64},
number = {2},
year = {1990},
pages = {127--130}
}

@article{DS,
title = {Strings in less than one dimension},
author = {M.R. Douglas and S.H. Shenker},
journal = {Nucl. Phys. B},
volume = {335},
number = {3},
year = {1990},
pages = {635--654}
}

@article{EEKT,
title = {Linear r-matrix algebra for classical separable systems},
author = {J.C. Eilbeck and V.Z. Enolsky and V.B. Kuznetsov and A.V. Tsyganov},
journal={J. Phys. A},
fjournal = {J. Phys. A: Math. Gen.},
volume = {27},
number = {2},
pages = {567--578},
doi = {10.1088/0305-4470/27/2/038},
year = {1994}
}

@article{PedroniVanhaeke,
title = {{A Lie algebraic generalization of the Mumford system, its symmetries and its multi-hamiltonian structure}},
author = {M. Pedroni and P. Vanhaecke},
journal={Regul. Chaotic Dyn.},
fjournal = {Regular and Chaotic Dynamics},
volume = {3},
number = {3},
pages = {132--160},
year = {1998}
}

@article{ClaeysVanlessen,
title = {{Universality of a Double Scaling Limit near Singular Edge Points in Random Matrix Models}},
author = {T. Claeys and M. Vanlessen},
journal={Commun. Math. Phys.},
volume = {273},
year = {2007},
pages = {499--532}
}

@inproceedings{ClaeysGrava,
place={Cambridge}, 
series={Mathematical Sciences Research Institute Publications}, 
title={{Critical asymptotic behavior for the Korteweg–de Vries equation and in random matrix theory}}, 
booktitle={Random Matrix Theory, Interacting Particle Systems, and Integrable Systems}, 
publisher={Cambridge University Press}, 
author={T. Claeys and T. Grava}, 
editor={P. Deift and P. Forrester}, 
year={2014},
pages={71--92}, 
collection={Mathematical Sciences Research Institute Publications}
}

@article{ClaeysItsKrasovsky,
title = {{Higher-order analogues of the Tracy-Widom distribution and the Painlev\'{e} II hierarchy}},
author = {T. Claeys and I. Krasovsky and A. Its},
journal = {Comm. Pure Appl. Math.},
volume = {63},
number = {3},
year = {2010},
pages = {362--412}
}

@article{DK0,
title={{Painlev\'{e} I asymptotics for orthogonal polynomials with respect to a varying quartic weight}},
author={M. Duits and A.B.J. Kuijlaars},
journal={Nonlinearity},
volume={19},
year={2006},
pages={2211--2245}
}

@article{FlaschkaNewell1980,
  author  = {H. Flaschka and A.C. Newell},
  title   = {Monodromy- and spectrum-preserving deformations. {I}},
  fjournal = {Communications in Mathematical Physics},
  journal={Commun. Math. Phys.},
  volume  = {76},
  year    = {1980},
  pages   = {65--116},
  doi     = {10.1007/BF01208204}
}

@article{Joshi1995,
  author  = {N. Joshi},
  title   = {The second {Painlev\'e} hierarchy and the stationary
             {KdV} hierarchy},
  journal = {Publ. Res. Inst. Math. Sci.},
  volume  = {40},
  number  = {3},
  pages   = {1039--1061},
  year    = {2004},
  doi     = {10.2977/PRIMS/1145475502}
}

@article{GordoaPickering1999,
  author  = {P.R. Gordoa and A. Pickering},
  title   = {Nonisospectral scattering problems: a key to integrable hierarchies},
  fjournal = {Journal of Mathematical Physics},
  journal={J. Math. Phys.},
  volume  = {40},
  year    = {1999},
  pages   = {5749--5786},
  doi     = {10.1063/1.533055}
}

@phdthesis{TarriconeThesis2022,
  author = {S. Tarricone},
  title  = {The {Painlev\'e}~{II} hierarchy: geometry and applications},
  school = {Concordia University  and Universit\'e d’Angers},
  year   = {2022},
}

\end{document}